\documentclass[journal,draftclsnofoot,onecolumn,12pt,twoside]{IEEEtran}
\IEEEoverridecommandlockouts
\usepackage{xcolor}
\usepackage{tikz}
\usetikzlibrary{spy,arrows}
\usepackage{pgfplots}
\usepackage[linesnumbered,ruled,vlined]{algorithm2e}
\SetKwInput{kwInit}{Initialization}
\usepackage{amsmath,epsfig,amssymb,amsthm,cite,url,xcolor}
\usepackage{hyperref}
\usepackage{enumerate}
\usepackage[latin1]{inputenc}
\DeclareMathOperator{\tr}{tr}

\DeclareMathOperator{\var}{Var} 
\DeclareMathOperator{\cov}{Cov} 
\theoremstyle{plain}
\newtheorem{theorem}{Theorem}
\newtheorem{assumption}{Assumption}

\theoremstyle{definition}

\newtheoremstyle{specialcasestyle}{1mm}{1mm}{\upshape}{}{\bfseries\upshape}{.}{0mm}{}
\theoremstyle{specialcasestyle}

\newcommand{\figref}[1]{Fig.~\protect\ref{#1}}

\newcommand{\ex}{{\mathbb E}}

\newcommand{\dto}{\overset{d}\longrightarrow }
\newcommand{\pto}{\overset{P}\longrightarrow }
 
\makeatletter
\def\endthebibliography{%
  \def\@noitemerr{\@latex@warning{Empty `thebibliography' environment}}%
  \endlist
}
\makeatother
\usepackage{graphicx}
\usepackage{placeins}
\usepackage{float}
\usepackage{tabularx}
\usepackage{tkz-euclide,subfigure}
\usepackage{epsfig,amssymb}
\usepackage{amsmath}
\usepackage{bbm}
\usetikzlibrary{intersections}
\usepgfplotslibrary{fillbetween}

\usepackage{mathtools}
\pgfplotsset{compat=1.12}
\usepackage[detect-all]{siunitx}
\usepackage{etoolbox}
\usepackage{svg}
\DeclareMathOperator*{\argmin}{arg\,min}
\DeclareMathOperator*{\argmax}{arg\,max}

\DeclarePairedDelimiter{\norm}{\lVert}{\rVert} 
\begin{document}

\title{Semi-Supervised Learning via Cross-Prediction-Powered Inference for Wireless Systems}


\author{\IEEEauthorblockN{Houssem Sifaou and Osvaldo Simeone}
 \thanks{The authors are with the King's Communications, Learning $\text{\&}$ Information Processing (KCLIP) lab within the Centre for Intelligent Information Processing Systems (CIIPS), Department of Engineering, King's College London, WC2R 2LS London, U.K. (e-mail: houssem.sifaou@kcl.ac.uk; osvaldo.simeone@kcl.ac.uk).
    }
    
    \thanks{The work of H. Sifaou and O. Simeone was supported by the European Union's Horizon Europe project CENTRIC (101096379). O. Simeone was also supported by the Open Fellowships of the EPSRC (EP/W024101/1) by the EPSRC project (EP/X011852/1), and by Project REASON, a UK Government funded project under the Future Open Networks Research Challenge (FONRC) sponsored by the Department of Science Innovation and Technology (DSIT).}

}
\maketitle

\begin{abstract}
In many wireless application scenarios, acquiring labeled data can be prohibitively costly, requiring complex optimization processes or measurement campaigns. 
Semi-supervised learning leverages unlabeled samples to augment the available dataset by assigning synthetic labels obtained via machine learning (ML)-based predictions. However, treating the synthetic labels as true labels may yield worse-performing models as compared to models trained using only labeled data. Inspired by the recently developed prediction-powered inference (PPI) framework, this work investigates how to leverage the synthetic labels produced by an ML model, while accounting for the inherent bias concerning true labels. To this end, we first review PPI and its recent extensions, namely tuned PPI and cross-prediction-powered inference (CPPI). { Then, we introduce two novel variants of PPI. The first, referred to as tuned CPPI, provides CPPI with an additional degree of freedom in adapting to the quality of the ML-based labels. The second, meta-CPPI (MCPPI), extends tuned CPPI via the joint optimization of the ML labeling models and of the parameters of interest.} Finally, we showcase two applications of PPI-based techniques in wireless systems, namely beam alignment based on channel knowledge maps in millimeter-wave systems and received signal strength information-based indoor localization. Simulation results show the advantages of PPI-based techniques over conventional approaches that rely solely on labeled data or that apply standard pseudo-labeling strategies from semi-supervised learning. Furthermore, the proposed tuned CPPI method is observed to guarantee the best performance among all benchmark schemes, especially in the regime of limited labeled data.

% Exploiting unlabeled samples, often more accessible, through machine learning (ML) predictions can enhance data-driven decision-making. However, fully trusting labels produced by black-box ML models may have a reverse effect and yield a worse performance.
% Inspired by the recently developed prediction-powered inference (PPI) techniques, we investigate in this work how to judiciously exploit labels from ML models while considering their imperfections and inherent bias. We investigate various PPI variants, including PPI, tuned PPI, and cross-prediction-powered inference (CPPI), tailored for wireless systems. 
% Additionally, we introduce a novel variant, tuned CPPI, building upon CPPI but providing the flexibility to adapt as a function of the prediction quality on unlabeled data. We consider two applications in wireless systems: beam alignment in mmWave communication systems based on channel knowledge map and RSS-based indoor localization. 
    
\end{abstract}
\begin{IEEEkeywords}
Prediction-powered inference, semi-supervised learning, channel knowledge map, indoor localization
\end{IEEEkeywords}

\maketitle

\section{Introduction}
\label{Intro}
\subsection{Context and Motivation}
Next-generation wireless systems are expected to rely extensively on machine learning (ML) and data-driven decision-making~\cite{sun2019application, simeone2018very,mao2018deep}. Optimizing effective ML algorithms hinges on the availability of high-quality labeled data. However, obtaining labeled data is a challenging task in numerous wireless scenarios due to the need to run time-consuming optimizations~\cite{liao2022deep,shental2019machine} or to collect data via on-air transmission~\cite{6244790,zhang2023csi}. Semi-supervised learning via pseudo-labeling provides a promising alternative by leveraging synthetic labels produced by ML models for unlabeled data~\cite{yang2022survey,van2024generative,li2023exploiting,yoo2017semi,kim2019semi,nguyen2020supervised,nayebi2017semi,soares2023semi,camelo2019semi,chen2024neuromorphic}. However, predictions generated by  ML models may be of insufficient quality. Therefore, making reliable use of synthetic labels requires an additional effort to reduce the bias caused by the discrepancy between synthetic and real labels. This is the focus of this work.

% Recent advances in deep learning and generative AI techniques provide a promising alternative, which enables the production of large amounts of synthetic data and the annotation of unlabeled data 

Specifically, we consider a semi-supervised setting in which unlabeled data $\tilde X$ are abundant, while labels $Y$ are difficult to obtain. The goal is to learn a parameter vector $\theta$. As illustrated in \figref{ppi_Fig}, we assume the availability of an ML model $f(X)$ that can assign a synthetic label $ f(X)$ to any input $X$. Assuming the model $f(X)$ to be pretrained, conventional semi-supervised learning schemes would augment the dataset with the synthetic data $(\tilde X, f(\tilde X))$ obtained from the unlabeled data $\tilde X$ to estimate parameter $\theta$~\cite{lee2013pseudo}.

Note that we do not explore here semi-supervised methods that aim at extracting information from the covariate distribution, e.g., via generative models or via unsupervised pre-training; or that augment the data sets by generating new samples $(\tilde{X},Y)$ based on manipulation of existing inputs $X$ (see, e.g., \cite{raviv2023data}). Such methods are complementary to the pseudo-labeling formulation adopted here and may be potentially combined with it~\cite{angelopoulos2023prediction}.

\begin{figure}[t!]
    \centering    \includegraphics[scale = 0.61]{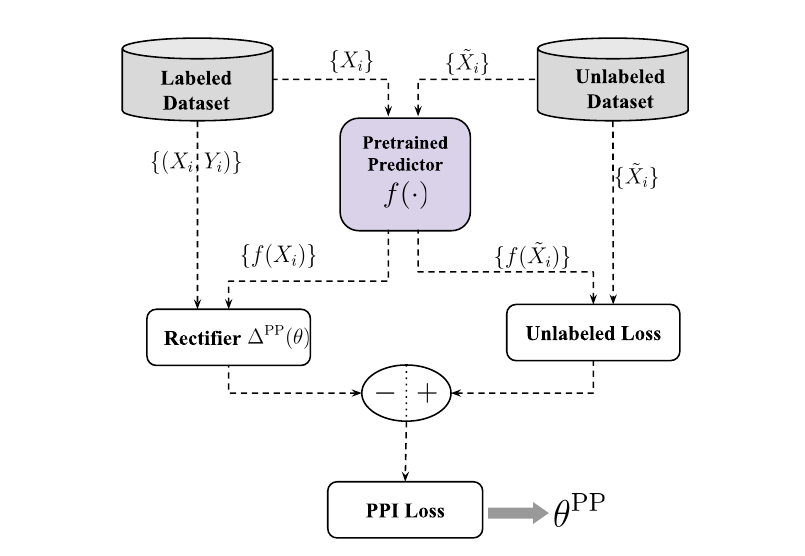} \vspace{-0mm}
    \caption{Illustration of the original PPI scheme \cite{angelopoulos2023prediction}: Using the labeled data and a pre-trained model $f(\cdot)$,  the rectifier term $\Delta^{\rm PP}(\theta)=\frac{1}{n}\sum_{i=1}^n \left[\ell_\theta(X_i,f(X_i))-\ell_\theta(X_i,Y_i)\right]$ is evaluated to estimate the prediction bias of the model $f(\cdot)$. This term is subtracted from the unlabeled loss $ \frac{1}{N}\sum_{i=1}^N \ell_\theta(\tilde X_i,f(\tilde X_i)) $, obtaining the PPI loss $L^{\rm PP}(\theta)$ in \eqref{loss_ppi_intro}.}
    \vspace{0mm}
    \label{ppi_Fig}
\end{figure}
% % Such a setting has been recently studied in~\cite{angelopoulos2023prediction,angelopoulos2023ppi,zrnic2023cross}, where a prediction-powered inference (PPI) framework has been developed. The framework allows reliable estimation and statistical inference, by leveraging predictions from ML models to annotate unlabeled data, for estimands of the form
% $$
% \theta^\star = \argmin_{\theta}\ex[\ell_\theta(X,Y)],
% $$
% where $\ell_\theta(X,Y)$ is a loss function. The approach provides a robust alternative to the conventional methods that rely solely on labeled data or fully trust the predictions on unlabeled data, treating them as ground-truth labels. 
The \emph{prediction-powered inference} (PPI) framework, introduced in~\cite{angelopoulos2023prediction} { (see also~\cite{dudik2014doubly,10.1214/16-STS589} for prior related art), takes a different approach by using a limited amount of labeled data $\{(X_i,Y_i)\}_{i=1}^n$ to estimate the bias $\Delta^{\rm PP}(\theta)$ caused by the mismatch between true labels $Y$ and predicted labels $f(\tilde X)$ as the average difference
\begin{align}
\Delta^{\rm PP}(\theta) = \frac{1}{n}\sum_{i=1}^n \left[\ell_\theta(X_i,f(X_i))-\ell_\theta(X_i,Y_i)\right],
\end{align}
where $\ell_\theta( X,Y)$ is a loss function dependent on the parameter $\theta$ under optimization.} Given an unlabeled dataset $\{\tilde X_i\}_{i=1}^N$ and the \emph{rectifier} term $\Delta^{\rm PP}(\theta)$, PPI estimates the population loss as 
\begin{align}
L^{\rm PP}(\theta) = \frac{1}{N}\sum_{i=1}^N \ell_\theta(\tilde X_i,f(\tilde X_i)) - \Delta^{\rm PP}(\theta),
\label{loss_ppi_intro}
\end{align}
The PPI approach is illustrated in \figref{ppi_Fig}.

The original PPI method may not necessarily improve over a baseline empirical risk minimization (ERM) that disregards unlabeled data. To obviate this issue, reference~\cite{angelopoulos2023ppi} introduced \emph{tuned PPI}, which adapts the use of unlabeled data depending on the quality of the predictions.
 
 Both PPI in~\cite{angelopoulos2023prediction} and tuned PPI in~\cite{angelopoulos2023ppi} assume the availability of a pretrained ML model $f(X)$ for annotating the unlabeled data. To alleviate this assumption, the authors in~\cite{zrnic2023cross} proposed \emph{cross-prediction-powered inference} (CPPI). In CPPI, the labeled data must be shared between the tasks of training the model $f(X)$ and evaluating a rectifier $\Delta^{}(\theta)$. Note that, unlike semi-supervised methods such as co-training techniques, this approach does not require the extraction of different ``views'' from the input, as done in image processing \cite{qiao2018deep}.

 To this end, as illustrated in \figref{cppi_Fig}, CPPI operates in a way similar to cross-validation techniques~\cite{stone1978cross,stone1974cross,cohen2024cross,park2023few}. Accordingly, $K$ models $\{f^{(k)}(X)\}_{k=1}^K$ are trained using different parts of the labeled dataset. These models are then used to annotate the unlabeled data and to obtain a rectifier for the prediction bias $\Delta^{\rm CP}(\theta)$ {  given by the average difference
\begin{align}
 \Delta^{\rm CP}(\theta)=\frac{1}{n}\sum_{k=1}^K \sum_{i\in\mathcal{D}^{(k)}}\left[\ell_\theta(X_i,f^{(k)}(X_i)) -\ell_\theta(X_i,Y_i)\right],\end{align}
where $\mathcal{D}^{(k)}$ is the part of the dataset not used to train model $f^{(k)}(\cdot)$}. This correction term is subtracted from the estimated loss $ \frac{1}{KN}\sum_{k=1}^K\sum_{i=1}^N \ell_\theta(\tilde X_i,f^{(k)}(\tilde X_i))$ to obtain the estimate
\begin{align}
L^{\rm CP}(\theta) = \frac{1}{KN}\sum_{k=1}^K\sum_{i=1}^N \ell_\theta(\tilde X_i,f^{(k)}(\tilde X_i)) - \Delta^{\rm CP}(\theta).
\label{cppi_loss_intro}
\end{align}
% which yields a better estimate with a lower variance compared to just using the labeled dataset when $N\gg n$ and model $f(\cdot)$ is sufficiently accurate\cite{angelopoulos2023prediction}. The PPI approach is illustrated in \figref{ppi_Fig}. In~\cite{angelopoulos2023ppi}, a tuned version of PPI allowing the flexibility to adapt depending on the quality of the predictions was introduced. Both PPI in~\cite{angelopoulos2023prediction} and tuned PPI in~\cite{angelopoulos2023ppi} assume the availability of a pretrained ML model for annotating the unlabeled data. To alleviate this assumption, the authors in~\cite{zrnic2023cross} proposed cross-prediction powered inference (CPPI), where the labeled data is leveraged for training $K$ prediction models, in a similar way as cross-validation techniques, and for removing the prediction bias as illustrated in \figref{cppi_Fig}.

In this work, we showcase application scenarios for the PPI framework in wireless systems. { We also introduce two new versions of CPPI. The first, called tuned CPPI, allows the flexibility to adapt the use of unlabeled data as a function of the quality of the trained prediction models $\{f^{(k)}(X)\}_{k=1}^K$. The second, referred to as meta-CPPI (MCPPI), integrates CPPI with meta pseudo-labeling (MPL)~\cite{pham2021meta}, a complementary approach that iteratively optimizes both the parameter of interest $\theta$ and the labeling model $f(\cdot)$.}

\begin{figure}[t!]
    \centering    \includegraphics[scale = 0.51]{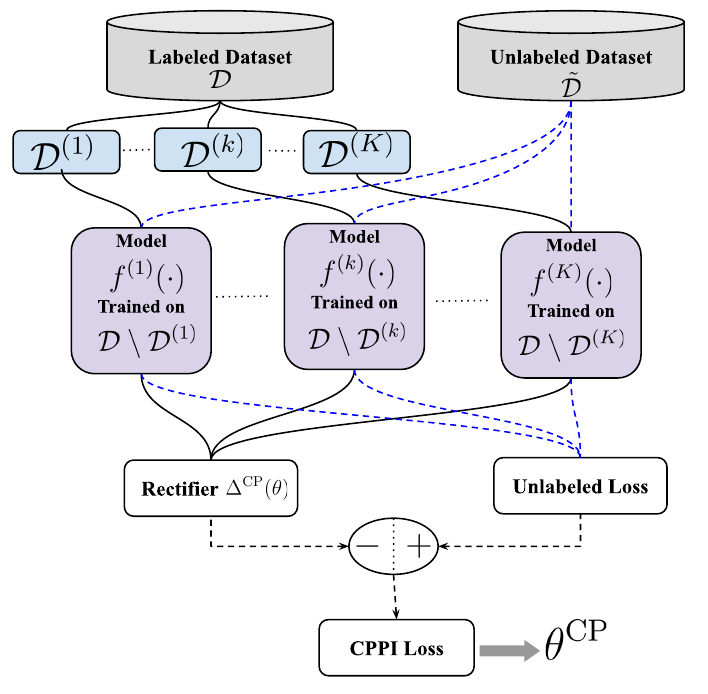} \vspace{-0.1cm}
    \caption{Illustration of the CPPI scheme \cite{zrnic2023cross}: The labeled data is divided into $K$ folds $\mathcal{D}^{(1)},\cdots,\mathcal{D}^{(K)}$, and $K$ prediction models are trained, with each model $f^{(k)}(\cdot)$ being trained on all labeled data except for fold $\mathcal{D}^{(k)}$. Using the $K$ trained models, a rectifier $\Delta^{\rm CP}(\theta)=\frac{1}{n}\sum_{k=1}^K \sum_{i\in\mathcal{D}^{(k)}}\left[\ell_\theta(X_i,f^{(k)}(X_i)) -\ell_\theta(X_i,Y_i)\right]$ is evaluated that estimates the prediction bias of models $\{f^{(k)}(\cdot)\}_{k=1}^K$. This term is subtracted from the unlabeled loss $ \frac{1}{KN}\sum_{k=1}^K\sum_{i=1}^N \ell_\theta(\tilde X_i,f^{(k)}(\tilde X_i)) $, obtaining the CPPI loss $L^{\rm CP}(\theta)$ in \eqref{cppi_loss_intro}.}
    \vspace{0mm}
    \label{cppi_Fig}
\end{figure}
 
\begin{figure*}[t!]
\centering 
\subfigure[Beam alignment]{
    \centering    \includegraphics[scale = 0.85]{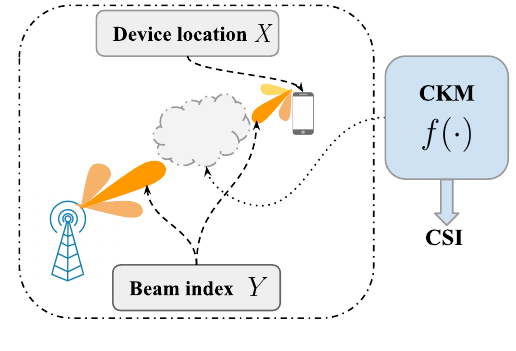} \vspace{-0.5cm}
    \vspace{0mm}
    \label{app_beam}}
    \subfigure[Indoor localization]{
    \centering    \includegraphics[scale = 0.83]{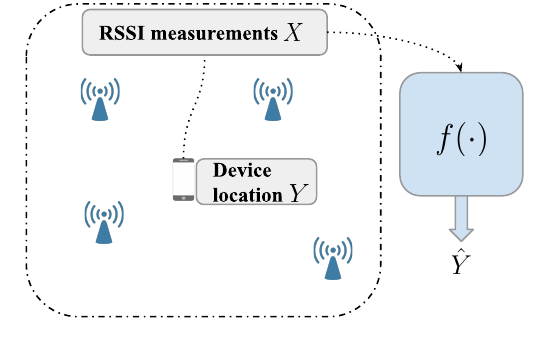} \vspace{-0.5cm}
    \vspace{0mm}
    \label{app_loc}}
     \caption{Illustration of the two application scenarios: (a) beam alignment in mmWave communication systems, in which the optimal beam index $Y$ is determined based on the device location $X$; and (b) an indoor localization system based on RSSI, in which the position of the device $Y$ is predicted based on RSSI measurements $X$ received from access points. In both cases, a pre-trained ML model $f(.)$ can be used to augment the labeled datasets, with a channel knowledge map (CKM) adopted for beam alignment.}
    \label{applications}
\end{figure*}

\subsection{Main Contributions}

The objective of this work is twofold. First, we study the benefits and potential application scenarios of the PPI framework for wireless communication systems. To this end, we first review PPI, tuned PPI, and CPPI, and then study two use cases in wireless systems, namely beam-alignment~\cite{zeng2021toward,wu2021environment} and localization~\cite{5290385,6244790,zhang2023csi}. Second, we propose a new variant of CPPI, referred to as tuned CPPI, that endows CPPI with the ability to adapt the use of unlabeled data to the quality of the ML prediction models, and thus to the amount of available labeled data.

To demonstrate the potential benefits of PPI and of the proposed tuned CPPI scheme for wireless systems, we study two application scenarios in which unlabeled data are more accessible than labeled data. As illustrated in Fig.~\ref{app_beam}, the first use case is beam alignment in millimeter-wave (mmWave) massive MIMO systems via channel knowledge maps (CKMs)~\cite{zeng2021toward,wu2021environment,heng2024site}. Beam alignment requires selecting the index $Y$ of the optimal beams within predefined codebooks given the location $X$ of the device being served. CKMs represent site-specific databases containing channel information at given locations of receivers and transmitters, such as channel state information (CSI) matrices or partial information regarding the strongest signal paths~\cite{zeng2021toward,wu2021environment}. By leveraging a CKM, one can thus produce an estimate $f(\tilde X)$ of the optimal beam for an unlabeled location $\tilde X$. The synthetic data $(\tilde X, f(\tilde X))$ can augment the labeled dataset $(X,Y)$ of locations $X$ and optimal beam indices $Y$, In this application, PPI and its extensions can be useful to correct the inevitable errors in the CSI provided by a CKM.

% Each base station (BS) or cluster of adjacent BSs is associated with its distinct CKM. CKMs simplify the process of acquiring real-time channel state information (CSI), potentially eliminating the need for explicit training~\cite{zeng2021toward}. However, to achieve the full merits of CKM a huge amount of labeled data is required, and acquiring such data can be costly and time-consuming. We propose to complement the labeled data with additional unlabeled samples, which consist of a set of new receiver locations drawn from the user location distribution in the coverage area. Such unlabeled samples are much easier to obtain since the user distribution could be estimated from previous data. Then, given both labeled and unlabeled data a softmax regression model, which maps a user location to the optimal beam index, is optimized using the proposed tuned CPPI technique. 

As a second application scenario, as shown in Fig.~\ref{app_loc}, we consider received signal strength information (RSSI)-based indoor localization. In RSSI-based positioning, RSSI measurements $X$ at a number of access points (APs) are used to infer position $Y$ of a mobile user~\cite{5290385,6244790}. The accuracy of this positioning technique depends on the size of the available labeled dataset of pairs $(X,Y)$ of RSSI fingerprint $X$ and location $Y$. Additional unlabeled data consisting of RSSI measurements $\tilde X$ was shown to be potentially useful by using semi-supervised learning~\cite{li2023exploiting,yoo2017semi,qian2021supervised}. PPI schemes can be thus beneficial as means to further enhance the reliability of semi-supervised learning in this context.

To summarize, the main contributions of this work are outlined as follows.
\begin{itemize}
    \item We showcase applications of the PPI framework in wireless communication systems for scenarios in which acquiring labeled data is costly, presenting PPI as an alternative to conventional semi-supervised learning. Specifically, we highlight beam alignment via CKM and RSSI-based indoor localization as use cases.
    \item We introduce tuned CPPI, a novel extension of CPPI \cite{zrnic2023cross} that flexibly adapts the degree of reliance on unlabeled data depending on the quality of the ML-based predicted labels.
    \item The performance of the proposed tuned CPPI and existing PPI benchmarks is tested for the mentioned use cases, as well as for mean estimation and linear regression problems. Numerical results demonstrate the superiority of the proposed scheme over all the benchmark schemes, including conventional semi-supervised learning.
    { \item We investigate the integration of tuned CPPI with meta pseudo-labeling~\cite{pham2021meta}, resulting in a new SSL scheme, termed meta-CPPI (MCPPI). This new approach, supporting the joint optimization of labeling models and the parameter of interest, demonstrates superior performance over both tuned CPPI and MPL.}
\end{itemize}

\subsection{Organization}
The rest of the paper is structured as follows. The next section introduces the semi-supervised setting and defines the studied problem. In Section \ref{PPI}, we describe the PPI framework including its three variants, namely PPI, tuned PPI, and CPPI. In Section \ref{tuned_cppi}, the proposed tuned CPPI scheme is introduced, while Section V presents MCPPI. Section \ref{synthetic} provides simulation examples using synthetic data. Two applications to wireless systems, namely beam alignment for mmWave communication and RSSI-based indoor localization, are discussed in Sections \ref{beam_alignment} and \ref{loc_application}, respectively. Concluding remarks are drawn in Section \ref{conclusion}.

\subsection{Notations}
Given a vector $x$, $x^T$ and $x^H$ denote the transpose and conjugate of $x$. The trace of a matrix $X$ is denoted by $\tr(X)$. $|\mathcal{S}|$ is used to denote the cardinality of set $\mathcal{S}$. $\var(X)$ denotes the covariance matrix a random vector $X$, i.e., $\var(X) = \ex [(X-\ex[X])(X-\ex[X])^T]$; while for random vectors $X$ and $Y$, $\cov(X,Y)$ denotes the cross-covariance matrix, i.e., $\cov(X,Y) = \ex [(X-\ex[X])(Y-\ex[Y])^T]$. For a sequence of random variables, $\{X_n\}_{n\in\mathbb{N}}$, we denote by $X_n\pto X$ (resp. $X_n\dto X$) the convergence in probability (resp. in distribution) to the random variable $X$ in the limit $n \to \infty$. 
% For two random sequences $\{X_n\}_{n\in\mathbb{N}}$ and $\{Y_n\}_{n\in\mathbb{N}}$, the notation $X_n = Y_n + o_P(1)$ means that $X_n - Y_n \pto 0$. Lastly, we denote by $X_n\overset{L^1}\to X$ the convergence in mean, that is, $\lim_{n\to \infty}\ex[|X_n-X|] \to 0$

\section{Problem definition}
\label{problem_def}

Consider a scenario in which acquiring high-quality labeled data is costly, while unlabeled samples are abundant. Specifically, a \emph{labeled} dataset $\mathcal{D} = \left\{(X_i, Y_i)\right\}^n_{i=1}$ of $n$ independent and identically distributed (i.i.d.) samples from an unknown distribution $ {P}_{XY} = {P}_X \times {P}_{Y|X}$ is available, along with an \emph{unlabeled} dataset $\tilde{\mathcal{D} }= \{\tilde X_i\}^N_{i=1}$ of $N$ i.i.d. samples drawn from the marginal distribution ${P}_X $. The unlabeled dataset is typically much larger than the labeled dataset, i.e., $N\gg n$. Examples of this setting include medical applications requiring experts' annotations~\cite{jiao2023learning,gao2021detection} and engineering applications, such as wireless systems, in which labels may require the execution of costly optimizations on real-world data~\cite{simeone2018very,yoo2017semi,camelo2019semi,qian2021supervised}.

Given a \emph{convex} loss function $\ell_\theta(X,Y)$, the objective is to reliably estimate the minimizer $\theta^\star \in \mathbb{R}^d$ of the \emph{population loss} $L(\theta)$, i.e.,
  \begin{align}
\theta^\star = \argmin_\theta L(\theta),\ \ {\rm with} \ \  L(\theta) = \ex[\ell_\theta(X,Y)],
\label{estimand}
\end{align}
where the expected value is taken over the unknown joint distribution ${P}_{XY}$. The formulation in \eqref{estimand} encompasses for instance mean estimation, quantile regression, and generalized linear models.

An unbiased estimator of the population loss in \eqref{estimand} can be obtained by using only the labeled data, yielding the classical \emph{empirical risk minimization} (ERM) estimator  
% where $\ell_\theta$ is a convex function in $\theta$. The formulation in \eqref{estimand} encompasses several inference problems such as estimating means, quantiles, and linear and logistic regression coefficients. For instance, the mean estimation problem is given by
%   \begin{align}
% \theta^\star = \ex [Y] = \argmin_\theta \ex \left[ (Y-\theta)^2\right],
% \label{mean_estimand}
% \end{align}
% The main objective is to obtain a valid and reliable estimate for $\theta^\star$ in \eqref{estimand} by leveraging both the labeled and the unlabeled datasets. Note that an unbiased estimator can be obtained using the labeled data as follows
  \begin{align}
\theta^{\rm ERM} = \argmin_\theta L^{\rm ERM}(\theta), 
\label{classic}
\end{align}
with
\begin{align}
 L^{\rm ERM}(\theta) = \frac{1}{n}\sum_{i=1}^n \ell_\theta(X_i, Y_i).
\label{classic_loss}
\end{align}

However, the ERM solution $\theta^{\rm ERM}$ can have a high variance when the labeled dataset is of limited size. It is thus of interest to leverage also the unlabeled data in order to obtain a better estimate for the population-optimal parameter $\theta^\star$.

To this end, assume the availability of a \emph{model} $f(X)$ that provides an estimate of the label $Y$. The model $f(X)$ can be a pretrained ML model or any other given predictor. Conventional \emph{semi-supervised learning} addresses the problem \cite{books/mit/06/CSZ2006}
  \begin{align}
  \theta^{\rm SS} = \argmin_\theta L^{\rm SS}(\theta),
\label{SS_est}
\end{align}
with 
\begin{align}
L^{\rm SS}(\theta) = \frac{n}{n+N} L^{\rm ERM}(\theta) + \frac{\gamma}{n+N}\sum_{i=1}^N \ell_\theta(\tilde X_i, f(\tilde X_i)),
\label{SS_loss}
\end{align}
where the loss $ L^{\rm SS}(\theta)$ averages over both labeled and unlabeled data, and $\gamma\geq0$ is a hyperparameter dictating the confidence in the synthetic labels  ~\cite{books/mit/06/CSZ2006,nabati2020using,camelo2019semi}. Note that with $\gamma=1$ synthetic and real labels are treated on an equal footing.

 Accordingly, conventional semi-supervised learning optimizes a biased estimate of the population loss in \eqref{estimand}, and the bias may cause significant performance degradation when the model $f(X)$ is not sufficiently accurate~\cite{angelopoulos2023prediction,angelopoulos2023ppi,zrnic2023cross}.

% Leveraging the unlabeled data using prediction-powered inference \cite{angelopoulos2023prediction} or cross-prediction \cite{zrnic2023cross} can provide more reliable estimators.

\section{Prediction-Powered Inference}
\label{PPI}
\emph{Prediction-powered inference} (PPI)~\cite{angelopoulos2023prediction} and its variants~\cite{angelopoulos2023ppi,zrnic2023cross} provide principled alternatives to the conventional semi-supervised estimator \eqref{SS_est}, which have been shown to have the desirable theoretical guarantees and empirical performance. This section reviews PPI~\cite{angelopoulos2023prediction}, tuned PPI~\cite{angelopoulos2023ppi}, and cross-PPI~\cite{zrnic2023cross}, providing the necessary background for the introduction of tuned cross-PPI in the next section.

\subsection{Prediction-Powered Inference}
In the setting described in the previous section, PPI uses the labeled data to quantify, and compensate for, the prediction bias of the model $f(X)$ as compared to the ground-truth labels. To this end, PPI addresses the problem~\cite{angelopoulos2023prediction}
\begin{align}
{\theta}^{\rm PP} =  \argmin_\theta L^{\rm PP}(\theta),
\label{ppi}
\end{align}
with cost function
\begin{align}
 L^{\rm PP}(\theta) &=\frac{1}{N}\sum_{i=1}^N \ell_\theta (\tilde X_i,f(\tilde X_i ) )  - \left[\frac{1}{n}\sum_{i=1}^{n}\ell_\theta( X_i,f( X_i ))- L^{\rm ERM}(\theta)\right].
 \label{ppi_loss}
\end{align}
Unlike the conventional estimate $L^{\rm SS}(\theta)$ in \eqref{SS_est}, the loss function $L^{\rm PP}(\theta)$ is an \emph{unbiased} estimate of the population loss $L(\theta)$ in \eqref{estimand}. In fact, the expected values of the first and second terms in \eqref{ppi_loss} are equal, canceling each other. 

The intuition behind the loss function \eqref{ppi_loss} is that the term in the square brackets serves as a \emph{rectifier} for the bias caused by the use of the model $f(X)$ to assign labels in the unlabeled loss $
\sum_{i=1}^N \ell_\theta (\tilde X_i,f(\tilde X_i ) )/N$. In fact, the rectifier term measures the error in the estimate of the loss on the labeled data points. The PPI objective \eqref{ppi_loss} subtracts the estimated error from the unlabeled loss, making the loss estimate $ \eqref{ppi_loss}$ unbiased.

Under the assumption that the loss function $\ell_\theta(X,Y)$ is a convex function of parameter $\theta$, reference \cite{angelopoulos2023prediction} demonstrated that the variance of the estimate $\hat{\theta}^{\rm PP}$ is lower than that of the ERM estimator \eqref{classic} as long as the model $f(X)$ is sufficiently accurate. { Reference \cite{angelopoulos2023prediction} also provided confidence sets for the optimal solution $\theta^\star$ using the PPI estimate \eqref{ppi}, as well as generalizations that address covariate distributional shifts between training and testing data. Extensions of PPI to neural networks are discussed in \cite{NEURIPS2023_819f4269} (see Sec. V).}

\subsection{Tuned Prediction-Powered Inference}
PPI is not guaranteed to improve over ERM when the model $f(X)$ is not sufficiently accurate. To address this issue, reference \cite{angelopoulos2023ppi} proposed \emph{tuned PPI}, a variant of PPI that automatically adapts to the quality of the prediction model $f(X)$. Tuned PPI selects the parameter vector $\theta$ as 
\begin{align}
{\theta}^{\rm PP}_\lambda =  \argmin_\theta L^{\rm PP}_\lambda(\theta),
\label{tppi}
\end{align}
where the cost function is defined as
\begin{align}
 L^{\rm PP}_\lambda(\theta) &= L^{\rm ERM}(\theta) + \lambda \left[\frac{1}{N}\sum_{i=1}^N \ell_\theta (\tilde X_i,f(\tilde X_i ) ) - \frac{1}{n}\sum_{i=1}^n\ell_\theta( X_i,f( X_i ))\right],
 \label{tunedppi_loss}
\end{align}
and $\lambda\in [0,1]$ is a tuning parameter. 

As for PPI, the loss in \eqref{tunedppi_loss} is an unbiased estimate of the population loss $L(\theta)$ for any value $\lambda\in [0,1]$. By varying parameter $\lambda$, tuned PPI ranges from ERM, which is obtained for $\lambda=0 $, to PPI, which is recovered for $\lambda=1$. Reference \cite{angelopoulos2023ppi} proposed a procedure to set the parameter $\lambda$ with the aim of minimizing the variance of the estimate ${\theta}^{\rm PP}_\lambda$. This way, when the predictor $f(X)$ is inaccurate, tuned PPI can revert to the conventional ERM estimator by setting $\lambda = 0$ in \eqref{tunedppi_loss}.

% Hereafter, we review some theoretical results from \cite{angelopoulos2023ppi}.

% that depicts the trade-off between fully trusting the predictions of the model $f$ on the unlabeled dataset ($\lambda = 1$) or completely ignoring them ($\lambda = 0$). Setting $\lambda =0$ yields the classical estimator \eqref{classic} while $\lambda = 1 $ yields the PPI estimator \eqref{ppi_loss}. The loss in \eqref{tunedppi_loss} is unbiased for any $\lambda$.

\subsection{Cross-Prediction-Powered Inference}
\label{cp_sec}
PPI and tuned PPI assume the availability of a model $f(X)$. In practice, however, model $f(X)$ may have to be trained using labeled data. Therefore, the available labeled dataset must be shared between the task of obtaining the prediction model $f(X)$ and the task of estimating the parameter vector $\theta$ using \eqref{ppi} or \eqref{tppi}.

\emph{Cross-PPI} (CPPI) addresses this problem via \emph{cross-validation}, enabling the use of the entire labeled dataset for training model $f(X)$, as well as for estimating parameter $\theta$. As in cross-validation, the labeled dataset $\mathcal{D} = \left\{(X_i, Y_i)\right\}^n_{i=1}$ is divided into $K$ \emph{folds}, with the first fold including data points $(X_i, Y_i)$ with indices $i$ in $\mathcal{D}^{(1)} = \{1,\cdots,n/K\}$, the second fold including data points $(X_i, Y_i)$ with indices $i$ in $\mathcal{D}^{(2)} = \{n/K+1,\cdots,2n/K\}$, and so on for the other folds $\mathcal{D}^{(3)},\cdots,\mathcal{D}^{(K)}$. For each $k = 1,\dots,K$, a model $f^{(k)}(X)$ is trained on all folds except for fold $\mathcal{D}^{(k)}$. As detailed next, the predictions of the $K$ models $\{f^{(k)}(X)\}_{k=1}^K$ are used to estimate the parameter vector $\theta^\star$ based on both labeled and unlabeled data. 
% Building on PPI, \cite{zrnic2023cross} proposed cross-PPI, which alleviates the assumption of the availability of a pretrained prediction model $f$. The promise is that, 
% The approach consists of training $K$ models $\{f^{(k)}\}_{k=1}^K$ using the labeled dataset $\mathcal{D} = \left\{X_i, Y_i\right\}^n_{i=1}$. Specifically, the latter is divided into $K$ folds, $I_1 = \{1,\cdots,n/K\}$, $I_2 = \{n/K+1,\cdots,2n/K\}$, etc. Each model $f^{(k)}$ is trained on all folds except fold $I_k$. 

Specifically, the cross-PPI estimate is obtained as
\begin{align}
{\theta}^{\rm CP} =  \argmin_\theta L^{\rm CP}(\theta),
\label{cp}
\end{align}
with cost function
\begin{align}
 L^{\rm CP}(\theta) &= \frac{1}{KN}\sum_{k=1}^K\sum_{i=1}^N \ell_\theta (\tilde X_i,f^{(k)}(\tilde X_i ) )- \left[\frac{1}{n}\sum_{k=1}^K\sum_{i \in \mathcal{D}^{(k)}}\ell_\theta( X_i,f^{(k)}( X_i ))-L^{\rm ERM}(\theta)\right].
 \label{cp_loss}
\end{align}
The first term in \eqref{cp_loss} is the empirical loss that uses the predictions of the $K$ models on the unlabeled data, while the second, rectifier, term corrects the bias caused by the use of the trained models $\{f^{(k)}(X)\}_{k=1}^K$ in the first term. In this regard, note that the correction term for model $f^{(k)}(X)$ in \eqref{cp_loss} is obtained by using data in fold $\mathcal{D}^{(k)}$, which is independent of the data model $f^{(k)}(X)$ was trained on. As a result, the CPPI loss $L^{\rm CP}(\theta)$ is an unbiased estimate of the population loss $L^{}(\theta)$. {  It is noted that CPPI, and PPI as a special case, increase the complexity of conventional supervised learning methods by requiring additional training of models $\{f^{(k)}(\cdot)\}_{k=1}^K$. }

% In \cite{angelopoulos2023ppi} confidence intervals that are more efficient to compute are obtained based on asymptotic normality.

\section{Tuned cross-prediction-powered inference}
\label{tuned_cppi}
The quality of the CPPI estimate \eqref{cp} depends on the accuracy of the trained models $\{f^{(k)}(X)\}_{k=1}^K$ in \eqref{cp_loss}. Therefore, when the trained models are not sufficiently accurate, CPPI is not guaranteed to improve over ERM, which uses only labeled data. Inspired by tuned PPI, in this section we introduce \emph{tuned CPPI}, which provides the flexibility to judiciously adapt the use of unlabeled data as a function of the quality of the trained models $\{f^{(k)}(X)\}_{k=1}^K$.

% to propose a novel estimator,  which uses a similar approach to cross-prediction by training $K$ models on the labeled data, while providing the flexibility to adapt depending on the quality of the trained models. 
% The proposed estimator, referred to as \emph{tuned cross-prediction}, is studied theoretically 

\subsection{Tuned Cross-Prediction-Powered Inference}
% Since the quality of the trained models is not guaranteed, the performance of the cross-prediction estimator \eqref{cp} can be worse than the classical estimator when the models are inaccurate.
In a manner similar to the tuned CPPI loss in \eqref{tcploss}, tuned CPPI introduces a tuning parameter $\lambda $ in the CPPI loss \eqref{cp_loss} so as to determine the degree of reliance on the unlabeled data depending on the quality of the trained models. Specifically, the proposed tuned CPPI estimator is given by
\begin{align}
{\theta}^{\rm CP}_\lambda =  \argmin_\theta L^{\rm CP}_\lambda(\theta),
\label{tcp}
\end{align}
where the loss function is 
\begin{align}
L^{\rm CP}_\lambda(\theta)  & = L^{ERM}(\theta)  + \lambda \left[\frac{1}{KN}\sum_{k=1}^K\sum_{i=1}^N \ell_\theta (\tilde X_i,f^{(k)}(\tilde X_i))  - \frac{1}{n}\sum_{k=1}^K\sum_{i\in \mathcal{D}^{(k)}} \ell_\theta (X_i,f^{(k)}(X_i))\right].
    \label{tcploss}
\end{align}
The tuned CPPI loss $L^{\rm CP}_\lambda(\theta)$ reduces to the CPPI loss in \eqref{cp_loss} when $\lambda=1$ and it recovers the ERM loss \eqref{classic} when $\lambda=0$. 

As for CPPI and ERM, the tuned CPPI loss in \eqref{tcploss} is an unbiased estimate of the population loss $L(\theta)$ for any $\lambda \in [0,1]$, i.e., 
\begin{align}
\ex [L^{\rm CP}_\lambda(\theta)]= L(\theta),
\label{unbiased}
\end{align}
where the average is taken over labeled and unlabeled data. Unlike CPPI, tuned CPPI offers the flexibility to tune the parameter  $\lambda \in [0,1]$ as a function of the quality of the trained models. This is done with the aim of minimizing the mean squared error (MSE) of the estimate ${\theta}^{\rm CP}_\lambda$ in \eqref{tcp}. That is, the parameter $\lambda$ is ideally chosen as the minimizer 
\begin{align}
\lambda^\star = \argmin_{\lambda} \text{MSE}(\lambda), \text{   with   } \text{MSE}(\lambda)= \ex[\norm{{\theta}^{\rm CP}_\lambda - \theta^\star}^2_2],
\label{lam_star}
\end{align}
with average evaluated over labeled and unlabeled data. The MSE in \eqref{lam_star} is estimated in practice from the data, yielding a data-dependent estimate $\hat\lambda_n$ of the ideal tuning parameter $\lambda^\star$. To this end, the next section derives an explicit expression for the optimal parameter $\lambda^\star$, which is then approximated by using data to obtain the estimate $\hat\lambda_n$.
% which is true since $\ex \left[ \ell_\theta \left(\tilde X_{i'},f^{(k)}(\tilde X_{i'})\right)\right]= \ex \left[\ell_\theta \left((X_i,f^{(k)}(X_i)\right)\right]$ for all $i'\in\{1,\cdots,N\}$, $k \in\{1,\cdots,K\}$, and $i\in D^{(k)}$. The tuning parameter $\lambda \in [0,1]$ can be optimized yielding a data-depending and prediction-quality-dependant tuning parameter $\hat\lambda$. We will discuss this in detail in the next subsection.
% For notational convenience, we let $\hat\lambda$ the value chosen by the user for $\lambda$, note that $\hat\lambda$ can be data dependent.

\subsection{Optimal Tuning Parameter}
\label{asymptotics}
In this subsection, we aim at deriving an explicit expression for the optimal tuning parameter $\lambda^\star$ in \eqref{lam_star}.
% following \cite{angelopoulos2023ppi} and \cite{zrnic2023cross}, we establish the asymptotic performance of the proposed tuned CPPI estimator in \eqref{tcp}. 
% The analysis is valid under the asymptotic regime defined in the following assumption.
% \begin{assumption} We assume that
% \begin{itemize}
% \item $n,N \to \infty$ with $\frac{n}{N}\to r $.
% \item $K$ is finite.
% \end{itemize}
% \label{assump1}
% \end{assumption}
To start, we note that the trained models $\{f^{(k)}(X)\}_{k=1}^K$ are identically distributed, since they are trained on identically distributed data. Define the average predictor as 
\begin{align}
\bar f(\cdot) = \ex \left[f^{(1)}(\cdot)\right] = \cdots = \ex \left[f^{(K)}(\cdot)\right],
\label{average_predictor}
\end{align}
where the expectation is taken over the labeled data used to train the models $\{f^{(k)}(X)\}_{k=1}^K$. The average model $\bar f(x)$ can be interpreted as the predictor obtained by training many models on independent datasets of size $n-n/K$, and then averaging their predictions. The analysis of the mean squared error $\text{MSE}(\lambda)$ in \eqref{lam_star} is based on the following assumption, which formalizes the property that the gradient of the population loss $\nabla\ell_\theta ( X,f^{(k)}( X))$ does not depend too strongly on the model index $k$ as the number of labeled data points, $n$, increases.

\begin{assumption} For input $X\sim P_X$ independent of model $f^{(k)}(\cdot)$, denote by $\var(\nabla\ell_\theta (X,f^{(k)}( X)) - \nabla\ell_\theta( X,\bar f ( X))|f^{(k)})$ the covariance matrix of the random vector $\nabla\ell_\theta (X,f^{(k)}( X)) - \nabla\ell_\theta( X,\bar f ( X))$ for a fixed model $f^{(k)}(\cdot)$. The square root of the entries of $\var(\nabla\ell_\theta (X,f^{(k)}( X)) - \nabla\ell_\theta( X,\bar f ( X))|f^{(k)})$ are assumed to converge in mean to zero, i.e.,
\begin{align}
&\ex\left[\!\sqrt{\var\left( \nabla\ell_\theta ( X,f^{(k)}( X))\!-\! \nabla\ell_\theta( X,\bar f ( X)) |f^{(k)}\!\right)}\! \right]\!\underset{n\to \infty}\longrightarrow 0,\nonumber\\ & \text{for any } k\in \{1,\cdots,K\},
\label{assump_eq}
\end{align}
where the outer expectation is taken over the distribution of the model $f^{(k)}(\cdot)$. The square root and the convergence in \eqref{assump_eq} are applied entry-wise.
% where $Var(\nabla\ell_\theta (X,f^{(k)}( X)) - \nabla\ell_\theta( X,\bar f ( X))|f^{(k)})$ denotes here the covariance matrix of vector $\nabla\ell_\theta ( X,f^{(k)}( X)) - \nabla\ell_\theta( X,\bar f ( X))$ with respect to the conditional distribution of $X$ given $f^{(k)}$, while $\overset{L^1}\to$ stands for the convergence in mean (with respect to the distribution of $f^{(k)}$) element-wise. The square root here is applied element-wise to the entries of the covariance matrix $\var\left( \nabla\ell_\theta ( X,f^{(k)}( X)) - \nabla\ell_\theta( X,\bar f ( X)) |f^{(k)}\right)$. 
\label{assump2}
\end{assumption}

Generalizing\cite[Theorem 2]{zrnic2023cross} and \cite[Theorem 1]{angelopoulos2023ppi}, the following theorem establishes the asymptotic normality of the tuned CPPI estimator $\hat\theta^{\rm CP}_{\lambda}$ in \eqref{tcp}. This result will be then used to evaluate the optimal tuning parameter $\lambda^\star$ in \eqref{lam_star}. 
To simplify notations, we define the  Hessian of the population loss as $H_\theta =\nabla^2L(\theta)$; we let $\nabla\ell_\theta = \nabla\ell_\theta (X,Y) $ be the gradient of the loss on a labeled data point $(X,Y)$; and we write $\nabla\ell_\theta^{\bar f} = \nabla\ell_\theta (X,\bar{f}(X)) $ for the gradient of the loss on a data point $X$ with label assigned by the average model $\bar f(\cdot)$ in \eqref{average_predictor}. Furthermore, we denote as
\begin{align}
V_{\bar f, \theta^\star}^\lambda  = \lambda^2\var\left( \nabla\ell_{\theta^\star}^{\bar f}\right)
\label{vf}
\end{align}
the covariance matrix of the gradient $\nabla\ell_\theta^{\bar f}$, and as 
\begin{align}
V_{\Delta, \theta^\star}^\lambda  = \var\left( \nabla\ell_{\theta^\star}^{} - \lambda \nabla\ell_{\theta^\star}^{\bar f}\right)
\label{vdelta}
\end{align}
the covariance matrix of $\nabla\ell_{\theta^\star}^{} - \lambda \nabla\ell_{\theta^\star}^{\bar f}$. 

% We also write as $\hat\lambda_n$ for the estimate of $\lambda^\star$ obtained using both labeled and unlabeled datasets.

\begin{theorem} For $n\to\infty$ with $n/N=r$, assume that  the estimate $\hat\lambda_n $ converges to some value $\lambda$, i.e., $\hat\lambda_n \pto \lambda$, and that the corresponding parameter $\theta^{\rm CP}_{\hat\lambda_n}$ in \eqref{tcp} converges to the optimal value $\theta^\star$, i.e., $\theta^{\rm CP}_{\hat\lambda_n} \pto \theta^\star$. Then, we have the limit
\begin{align}
    \sqrt{n}(\theta^{\rm CP}_{\hat\lambda_n} - \theta^\star) \underset{n \to \infty}{\dto }\mathcal{N}(0, \Sigma_\lambda),
    \label{main_result}
\end{align} 
with covariance matrix
\begin{align}
 \Sigma_\lambda = H_{\theta^\star}^{-1} \left( r \cdot V_{\bar f, \theta^\star}^\lambda +  V_{\Delta, \theta^\star}^\lambda  \right) H_{\theta^\star}^{-1}. 
    \label{sigma}
\end{align} 
% with $V_{\bar f, \theta^*}^\lambda  = \lambda^2\var\left( \nabla\ell_{\theta^*}^{\bar f}\right)$ and $V_{\Delta, \theta^*}^\lambda  = \lambda^2\var\left( \nabla\ell_{\theta^*}^{} - \lambda \nabla\ell_{\theta^*}^{\bar f}\right)$.
\label{thm1}
\end{theorem}
\begin{proof}The proof follows the same techniques used in \cite[Theorem 2]{zrnic2023cross} and \cite[Theorem 1]{angelopoulos2023ppi}, and is thus omitted.
\end{proof}
% Theorem \ref{thm1} is a generalization of \cite[Theorem 2]{zrnic2023cross} and \cite[Theorem 1]{angelopoulos2023ppi} to the case of tuned CPPI. Theorem \ref{thm1} establishes the asymptotic normality of the tuned CPPI estimator $\hat\theta^{\rm CP}_{\hat\lambda}$ which allows to compute confidence intervals for $\theta^\star$ as well as to obtain the optimal tuning parameter $\lambda^\star$. The main objective of this work is the quality of the point estimate without particular focus on confidence intervals, and thus the result of Theorem \ref{thm1} will be used for optimizing the tuning parameter $\lambda$.
% The following corollary is a direct consequence of Theorem \ref{thm1}.
% \begin{cor}
%     The mean squared error $\text{MSE}(\lambda)$ of the tuned CPPI estimate $\theta^{\rm CP}_{\hat\lambda}$ satisfies
%     \begin{align}
%         \text{MSE}(\lambda) -  \overline{\text{MSE}}(\lambda)  \to 0, \text{  as  } n\to \infty,
%         \label{MSE_conver}
%     \end{align}
%     where $\overline{\text{MSE}}(\lambda) = \frac{1}{n}\tr(\Sigma_\lambda )$.
% \end{cor}
We note that for a fixed $\lambda\in[0,1]$, the consistency $\theta^{\rm CP}_{\lambda} \pto \theta^\star$ holds if the loss function $L^{\rm CP}_\lambda(\theta)$ is convex in $\theta$ or the parameter space is compact\cite{angelopoulos2023prediction,angelopoulos2023ppi,zrnic2023cross,van2000asymptotic}. For instance, the convexity of loss $L^{\rm CP}_\lambda(\theta)$ holds for all generalized linear models. We refer the reader to \cite{angelopoulos2023ppi,van2000asymptotic} for more discussion.

By \eqref{main_result}, the asymptotic mean squared error, which is proportional to $\tr(\Sigma_\lambda)$, depends on the curvature of the population loss around the optimal value $\theta^\star$ via the Hessian $H_{\theta^\star}$; on the inherent variability of the average predictor $\bar f(X)$ via the term $V_{\bar f, \theta^\star}^\lambda$ in \eqref{vf}; and on the accuracy of the average predictor $\bar f(X)$ via the term $V_{\Delta, \theta^\star}^\lambda$ in \eqref{vdelta}.

With this result at hand, the optimal tuning parameter $\lambda^\star$ in \eqref{lam_star} can be evaluated, in the limit $n\to \infty$, as
\begin{align}
&\lambda^\star = \argmin_{\lambda} \tr (\Sigma_\lambda) \nonumber\\&= \frac{\tr\left(H_{\theta^\star}^{-1}\!\left( \cov(\nabla\ell_{\theta^\star},\nabla\ell^{\bar f}_{\theta^\star})\!+\!\cov(\nabla\ell^{\bar f}_{\theta^\star},\nabla\ell_{\theta^\star}) \right)\!H_{\theta^\star}^{-1}\right)}{2(1+r)\tr\left(H_{\theta^\star}^{-1}  \var(\nabla\ell^{\bar f}_{\theta^\star})H_{\theta^\star}^{-1}\right)},
\label{opt_lam}
\end{align}
where  $\cov(\nabla\ell_{\theta^\star},\nabla\ell^{\bar f}_{\theta^\star})$ represent the  cross-covariance matrix of vectors $\nabla\ell_{\theta^\star}$ and $\nabla\ell^{\bar f}_{\theta^\star}$ with respect to random variables $(X,Y)\sim P_{XY}$; and $\var(\nabla\ell^{\bar f}_{\theta^\star})$ is the covariance matrix of vector $\nabla\ell^{\bar f}_{\theta^\star}$ with respect to random variable $X\sim P_X$. To estimate the optimal parameter $\lambda^\star$ in practice,  we resort to bootstrapping techniques, as done in~\cite{zrnic2023cross} for estimating the variance of the CPPI estimator. This is described next.

% We emphasize that the consistency assumption in Theorem \ref{thm1}, i.e., $\theta^{\rm CP}_{\hat\lambda_n} \pto \theta^*$ holds as long as $\hat\lambda_n \pto \lambda$ holds and $L_\lambda(\theta)$ is convex under some additional mild assumptions.

% since the loss function $\ell_\theta(X,Y)$ is convex by assumption. Thus, for any fixed $\lambda\in[0,1]$, the consistency of the tuned CPPI estimate, i.e, $\theta^{\rm CP}_{\lambda} \pto \theta^\star
% $, holds~\cite{zrnic2023cross}.

\subsection{Estimating the Optimal Tuning Parameter}
The optimized tuning parameter $\lambda^\star$ in \eqref{opt_lam} requires the evaluation of the Hessian $H_{\theta^\star}$, of the covariance matrix $\var(\nabla\ell^{\bar f}_{\theta^\star})$ of vector $\nabla\ell^{\bar f}_{\theta^\star}$, and of the cross-covariance matrix $\cov(\nabla\ell_{\theta^\star},\nabla\ell^{\bar f}_{\theta^\star})$ between vectors $\nabla\ell_{\theta^\star}$ and $\nabla\ell^{\bar f}_{\theta^\star}$. These quantities depend on the optimal parameter vector $\theta^\star$ and on the average predictor $\bar{f}(X)$ in \eqref{average_predictor}, and they need to be estimated from data. To this end, as in~\cite{angelopoulos2023ppi}, we first fix an arbitrary value $\lambda\in [0,1]$ to obtain an estimate of the optimal parameter vector $\theta^\star$ via the solution $\theta^{\rm CP}_\lambda$ in \eqref{tcp}. This estimate, which we denote as $\hat\theta$, is consistent for convex loss functions $L^{\rm CP}_\lambda(\theta)$, as discussed in Sec. \ref{asymptotics}.
% we need first a consistent estimator for the latter and this can be any $\hat{\theta}= {\theta}_\lambda^{\rm CP} $ obtained from \eqref{tcp} by taking any fixed $\lambda \in [0,1]$~\cite{angelopoulos2023ppi}, since these estimates are consistent as per the result of Theorem \ref{thm1}.
The Hessian $H_{\theta^\star}$ is then estimated by using the empirical estimate
\begin{align}
\hat H_{\hat\theta} = \frac{1}{n}\sum_{i=1}^n \nabla^2 \ell_{\hat\theta^{}_{}}(X_i,Y_i),
\label{hessian_est}
\end{align}
 obtained from the labeled data. 
 
Then, as in~\cite{zrnic2023cross}, we apply bootstrapping to simulate several runs of the training process, obtaining an estimate of the average predictor $\bar f(X)$. The estimates of $\var(\nabla\ell^{\bar f}_{\theta^\star})$ and $\cov(\nabla\ell_{\theta^\star},\nabla\ell^{\bar f}_{\theta^\star})$, denoted by $\widehat{\var}(\nabla\ell^{\bar f}_{\theta^\star})$ and $\widehat{\cov}(\nabla\ell_{\theta^\star},\nabla\ell^{\bar f}_{\theta^\star})$, respectively, are obtained by using empirical estimates from labeled and unlabeled data as detailed in Appendix A.

Using the estimates $\hat H_{\hat\theta}$, $\widehat{\var}(\nabla\ell^{\bar f}_{\hat\theta}) $, and $\widehat{\cov}(\nabla\ell_{\hat\theta},\nabla\ell^{\bar f}_{\hat\theta})$, a plug-in estimator of the optimal parameter $\lambda^\star$ in \eqref{opt_lam} is finally obtained as
\begin{align}
\hat\lambda = \frac{\tr\left(\hat H_{\hat\theta}^{-1} \left( \widehat{\cov}(\nabla\ell_{\hat\theta^{}},\nabla\ell^{\bar f}_{\hat\theta})+\widehat{\cov}(\nabla\ell^{\bar f}_{\hat\theta},\nabla\ell_{\hat\theta}) \right) \hat H_{\hat\theta}^{-1}\right)}{2(1+n/N)\tr\left(\hat H_{\hat\theta}^{-1}  \widehat{\var}(\nabla\ell^{\bar f}_{\hat\theta^{}})\hat H_{\hat\theta^{}}^{-1}\right)}.
\label{opt_lam_hat}
\end{align}

The optimized parameter $\hat\lambda$ in \eqref{opt_lam_hat} can fall outside the interval $[0,1]$. To solve this issue, one can clip its value back to $[0,1]$ or apply one-step estimators \cite{van2000asymptotic,angelopoulos2023ppi}. We adopt the first option in our experiments.

Overall, obtaining the estimate $\hat\lambda$ in \eqref{opt_lam_hat} involves two main steps. First, an estimate $\hat\theta$ of the optimal parameter $\theta^\star$ is obtained by addressing problem \eqref{tcp} for a fixed $\lambda\in [0,1]$, and parameter $\hat \lambda$ is computed using \eqref{opt_lam_hat}. Then, the tuned CPPI estimator is obtained by addressing problem \eqref{tcp} using $\lambda =\hat \lambda$. Algorithm \ref{alg} outlines the key steps of the proposed scheme.

Finally, we note that for the case of mean estimation, where the objective is to estimate parameter $\theta = \ex[Y]$ and the loss function in \eqref{estimand} corresponds to $\ell_\theta(X,Y) = (Y-\theta)^2$, the optimal tuning parameter $\lambda^\star$ in \eqref{opt_lam} has the simplified expression \cite{angelopoulos2023ppi}
\begin{align}
\lambda^\star = \frac{\cov(Y, \bar f(X))}{(1+r) \var(\bar f(X))},
    \label{opt_lam_mean}
\end{align}
which can be estimated in a manner analogous to \eqref{opt_lam_hat}.

\begin{algorithm}[!t]
 \SetAlgoLined
    \KwIn{Labeled dataset $\mathcal{D}$,  unlabeled dataset $\tilde{\mathcal{D}}$, number of prediction models $K$.}
    \KwOut{Tuned CPPI parameter estimate.}
1. Use the labeled data to train the prediction models $\{f^{(k)}(X)\}^K_{k=1}$ using the procedure described in Section \ref{cp_sec}.\\
Fix $\lambda_1\in[0,1]$.\\
2. Compute the tuned CPPI loss in \eqref{tcploss} for $\lambda = \lambda_1$ using the labeled and unlabeled data and solve \eqref{tcp} to obtain $\hat\theta = \theta^{\rm CP}_{\lambda_1}$.\\
3. Compute the Hessian estimate $\hat H_{\hat\theta} $ using \eqref{hessian_est}.\\
4. Compute $\widehat{\var}(\nabla\ell^{\bar f}_{\theta^\star})$ and $\widehat{\cov}(\nabla\ell_{\theta^\star},\nabla\ell^{\bar f}_{\theta^\star})$ using \eqref{bootstrap_cov_unlabeled} and \eqref{bootstrap_cross_cov}, respectively.\\
5. Obtain $\hat\lambda_n$ using \eqref{opt_lam_hat}.\\
6. Solve $\eqref{tcp}$ again for $\lambda = \hat\lambda_n$ to obtain $\theta^{\rm CP}_{\hat\lambda_n}$.
% \Return $\theta^{\rm CP}_{\hat\lambda_n}$\\
\caption{Tuned CPPI}
\label{alg}  \vspace{0mm}
\end{algorithm}

{ 
\section{Cross-Prediction Powered Inference in Practice}
\label{CPPI_nn}

This section first discusses the practical implementation of PPI variants for neural network models optimized using stochastic gradient descent (SGD) methods. Then, we present the integration of the proposed tuned CPPI with MPL~\cite{pham2021meta}, a strong SSL benchmark, resulting in a scheme potentially outperforming both tuned CPPI and MPL.
\subsection{PPI with Neural Networks}
\label{PPI_nn}
The recent work~\cite{NEURIPS2023_819f4269} has investigated the benefits of PPI when applied to neural networks. The approach consists of optimizing a slightly modified version of the PPI loss adapted to SGD algorithms. Specifically, at every gradient step $t$, the following loss is considered
\begin{align}
 L^{\rm PP}_t(\theta_t) &= \ex_{\tilde X \sim \tilde{\mathcal{D}}} \left[\ell_{\theta_t} (\tilde X,f(X) )\right]  -\kappa_t \ex_{ (X,Y) \sim {\mathcal{D}}} \left[\ell_{\theta_t}( X,f( X ))-  \ell_{\theta_t}( X,Y)\right],
 \label{ppi_loss_nn}
\end{align}
where the expectation $\ex_{\tilde X \sim \tilde{\mathcal{D}}}[\cdot]$ is estimated using a mini-batch drawn from the unlabeled dataset $ \tilde{\mathcal{D}}$, while the expectation $\ex_{ (X,Y) \sim {\mathcal{D}}} [\cdot]$ is estimated using a mini-batch from the labeled dataset $\mathcal{D}$. The scalar hyperparameter $k_t$ allows for the gradual introduction of the bias term in the training process in order to avoid numerical instabilities~\cite{NEURIPS2023_819f4269}. Choices for hyperparameter $k_t$ suggested in~\cite{NEURIPS2023_819f4269} are $t/T$ and $t^2/T^2$, with $T$ being the total number of gradient steps. 

The same type of loss can be applied to extend tuned CPPI to neural networks yielding the following loss at step $t$
\begin{align}
 & L^{\rm CP}_{\lambda, t}(\theta_t)  =   \lambda\sum_{k=1}^K \ex_{\tilde X \sim \tilde{\mathcal{D}}} \left[\ell_{\theta_t} (\tilde X,f^{(k)}(\tilde X)) \right] - {\kappa_t}\sum_{k=1}^K \ex_{ (X, Y) \sim {\mathcal{D}}^{(k)}} \left[\lambda  \ell_{\theta_t} (X,f^{(k)}(X)) -\ell_{\theta_t}( X,Y)\right],
    \label{tcploss_nn}
\end{align}
where the expectation $\ex_{ (X, Y) \sim {\mathcal{D}}^{(k)}}$ is estimated using a mini-batch from ${\mathcal{D}}^{(k)}$.
 % Numerical results validating the efficiency of the approach in \eqref{ppi_loss_nn} and \eqref{tcploss_nn} are provided in Section~\ref{beam_alignment}.

\subsection{Meta Pseudo-Labeling}
\label{CPPI_mpl}
MPL~\cite{pham2021meta} is a strong SSL benchmark that borrows ideas from meta-learning~\cite{finn2017model}. Unlike the classical pseudo-labeling approach~\cite{lee2013pseudo} where a fixed pre-trained ``teacher'' model $f(\cdot)$ generates pseudo-labels for a ``student'' model $S(\cdot)$ to learn from, MPL trains the teacher model $f(\cdot)$ along with the student model. A feedback term on the performance of the student model on the labeled data is added to the training loss of the teacher model $f(\cdot)$ to improve the quality of the generated pseudo-labels. 
% The teacher model plays the role of the model $f(\cdot)$ in the context of PPI and SS introduced in s

We start by introducing MPL, and then we propose a new SSL scheme combining tuned CPPI and MPL. We denote the student model explicitly as $S_\theta(\cdot)$, with $\theta$ being its parameter vector as in the rest of the paper. Furthermore, the teacher model is denoted by $f_\phi(\cdot)$, highlighting its dependence on the parameter vector $\phi$. 

The student model parameter $\theta$ is trained to minimize the loss on the unlabeled dataset $\tilde{\mathcal{D}}$ with pseudo-labels generated by the teacher model $f_\phi(\cdot)$, that is,
\begin{align}
L_{}(\theta, \phi) = \ex_{\tilde X \sim \tilde{\mathcal{D}}} \left[\ell_\theta(\tilde X, f_\phi(\tilde X))\right],
\label{mpl_student_loss}
\end{align}
with expectation estimated as in \eqref{ppi_loss_nn}. 

For the teacher model objective, MPL introduces a new term accounting for the feedback of the student model on the quality of the pseudo-labels. Specifically, the teacher model objective is  
\begin{align}
\mathcal{L}(\phi) = \ex_{(X,Y) \sim \mathcal{D}} \left[\ell_\phi( X, Y)\right] + \ex_{(X,Y) \sim \mathcal{D}} \left[\ell_{\theta^\star(\phi)}( X, Y)\right],
\label{mpl_teacher_loss}
\end{align}
with
\begin{align}
\theta^\star(\phi) = \argmin_{\theta} L_{}(\theta, \phi),
\label{mpl_theta_star}
\end{align}
being the optimized student parameter given the teacher model parameter $\phi$. Note that the first term in \eqref{mpl_teacher_loss} is the classical objective of the teacher model on the labeled data as in ERM, while the second term accounts for the performance of the student on the labeled data. To simplify the optimization of \eqref{mpl_teacher_loss}, in a manner similar to MAML \cite{finn2017model}, the optimal parameter vector $\theta^\star(\phi)$ in \eqref{mpl_theta_star} is approximated with one step gradient update as
\begin{align}
\theta^\star(\phi) \leftarrow\theta - \eta_S\nabla_{\theta}L_{}(\theta, \phi),
\label{mpl_theta_approx}
\end{align}
where $\theta$ is the current iterate and $\eta_S$ is the learning rate of the student model.

At each iteration, the MPL algorithm performs a gradient step on the student model using \eqref{mpl_theta_approx} and then a gradient step on the teacher model using the objective \eqref{mpl_teacher_loss}. Computing the gradient of \eqref{mpl_teacher_loss} with respect to $\phi$ requires additional derivations as the second term depends on $\theta'(\phi)$. We refer the reader to the original reference~\cite{pham2021meta} for details.

\subsection{Meta-CPPI}
Inspired by MPL, we propose a new scheme that combines MPL with tuned CPPI. Tuned CPPI works in a similar fashion to classical pseudo-labeling where the teacher models $\{f^{(k)}(\cdot)\}_{k=1}^K$ are trained once, and they are kept fixed when optimizing the model parameter $\theta$. In the proposed meta-CPPI (MCPPI), these models are trained jointly with the student model $S_\theta(\cdot)$, while using the tuned CPPI loss in \eqref{tcp} as an objective for the student model. 

To elaborate, we denote the $K$ teacher models as $\{f^{(k)}_{\phi^k}(\cdot)\}_{k=1}^K$, with ${\phi^k}$ being the parameter vector of model $k$. The student model $S_\theta(\cdot)$ is trained based on the pseudo-labels generated by the teachers $\{f^{(k)}_{\phi^k}(\cdot)\}_{k=1}^K$ using the loss
\begin{align}
&L_{\lambda}^{\rm MCP}(\theta, \{\phi^k\}_{k=1}^K) =\lambda \sum_{k=1}^K\ex_{\tilde X \sim \tilde{\mathcal{D}}} \left[\ell_\theta(\tilde X, f^{(k)}_{\phi^k}(\tilde X))\right]
 - \sum_{k=1}^K\ex_{ (X,Y) \sim {\mathcal{D}}^{(k)}} \left[\lambda \ell_\theta( X, f^{(k)}_{\phi^k}(X)) - \ell_\theta( X, Y)\right].
\label{metacppi_student_loss}
\end{align}
Furthermore, each teacher model $k$ is trained to minimize the loss
\begin{align}
\mathcal{L}_{}^k(\phi^k) = \ex_{(X,Y) \sim \mathcal{D}} \left[\ell_{\phi^{k}}( X, Y) + \ell_{\theta^\star(\{\phi^k\}_{k=1}^K)}( X,Y)\right],
\label{metacppi_teacher_loss_approx}
\end{align}
where, as in MPL, the optimal student model parameters $\theta^\star(\{\phi^k\}_{k=1}^K)$ are approximated with one gradient update as 
\begin{align}
\theta^\star(\{\phi^k\}_{k=1}^K) \leftarrow \theta - \eta_S\nabla_{\theta}L_{ \lambda}^{\rm MCP}(\theta, \{\phi^k\}_{k=1}^K).
\label{metacppi_theta_approx}
\end{align}
Computing the gradient of the second term in \eqref{metacppi_teacher_loss_approx} with respect to $\phi^k$ requires additional derivations, which are provided in Appendix B along with the pseudo-code for the proposed MCPPI algorithm.}

% \begin{algorithm}[!t]
%  \SetAlgoNoLine
%     \KwIn{Labeled dataset $\mathcal{D}$,  unlabeled dataset $\tilde{\mathcal{D}}$}
%     \KwOut{Student model $S_\theta(\cdot)$.}
% \For{Gradient step g in $\{1,\cdots, G\}$}{
% Sample an unlabeled data batch $\tilde{X}$ and a labeled data batch $({X}, Y)$.\\
% Sample pseudo-labels $\hat{\tilde Y} \sim P(.|\tilde{X}, \phi) = T_\phi(\tilde{X})$ (assuming $T_\phi(.)$ return a probability distribution over the labels.

% }

% \caption{Meta pseudo-labeling}
% \label{mpl_alg}  \vspace{0mm}
% \end{algorithm}

\section{Mean Estimation and Linear Regression}
\label{synthetic}
In this section, we report the results of an experiment using synthetic data to provide a first performance comparison across all PPI schemes. Specifically, we consider two variants of the estimation problem in \eqref{estimand}, namely mean estimation, corresponding to the loss function $\ell_\theta(X,Y) = (Y-\theta)^2$ with $Y\in\mathbb{R}$ and $\theta\in\mathbb{R}$, and linear regression coefficients estimation, corresponding to the loss function $\ell_\theta(X,Y) = (Y- X^T\theta)^2$ with $X\in\mathbb{R}^d$, $Y\in\mathbb{R}$, and $\theta\in\mathbb{R}^d$. In all experiments in this section, we fix the size of the unlabeled dataset to $N=10000$ samples, {  and we vary the size of the labeled dataset $n$.}

We consider the following model for data generation as in~\cite{zrnic2023cross}
\begin{align}
Y = \mu + X^T\beta+z,
\label{data_gen_model}
\end{align}
where $\mu\in\mathbb{R}$ is a fixed constant, $X \in\mathbb{R}^d$ is distributed as $X\sim\mathcal{N}(0,I_d)$, $\beta = \frac{R\sigma}{\sqrt{2}}[1,\cdots,1]^T\in\mathbb{R}^d$ for some given positive constants $R$ and $\sigma$, and $z\sim\mathcal{N}(0,\sigma^2(1-R^2))$ is independent of $X$. Parameters $\mu$ and $\sigma$ are fixed in all experiments, while parameter $R$ is varied in the interval $[0,1]$. Parameter $R$ dictates the degree to which the output variable $Y$ can be explained through the feature vector $X$. In particular, when $R=0$, the label $Y$ is independent of $X$, while, when $R=1$, the label $Y$ is a deterministic function of $X$.

For the CPPI and tuned CPPI schemes, $K=5$ models are trained by following the cross-validation procedure in Sec. \ref{cp_sec}, where each model $f^{(k)}(X)$ is a random forest regression model~\cite{ho1995random}. For the PPI and tuned PPI schemes, as in~\cite{zrnic2023cross}, half of the labeled data is used to train a random forest regression model $f(X)$, and the rest is used for estimating $\theta$. The MSE between the true parameter vector and its estimate is used as a performance metric. We compare the performance of the proposed tuned CPPI scheme with the benchmark schemes ERM, SS (with $\gamma=1$), PPI, tuned PPI, and CPPI described in Sec. \ref{problem_def} and in Sec. \ref{PPI}.

\subsection{Mean Estimation}

\begin{figure*}[t]
\begin{center}
\subfigure[$R^2=0.25$]  
{  
\begin{tikzpicture}[scale=0.6]
\begin{axis}[
tick align=outside,
tick pos=left,
x grid style={white!69.0196078431373!black},
xlabel={Size of the labeled dataset $n$},
xmajorgrids,
xmin= 100, xmax=400,
xtick style={color=black},
y grid style={white!69.0196078431373!black},
ylabel={Mean squared error},
ymajorgrids,
ymin=0.009, ymax=0.04,
ytick style={color=black},
grid=major,
scaled ticks = true,
legend pos = south west,
legend style={nodes={scale=0.6, transform shape}},
grid style=densely dashed,
]

\addplot  [thick, color = darkgray, mark = pentagon*, mark size = 2, mark repeat = 0, mark phase = 0]
coordinates {
(100, 0.03831341298891767)(200, 0.019451322614850382)(300, 0.012559468263247395)(400, 0.010986593031607171)(500, 0.00929962971583781)
}; \addlegendentry{ERM}
\addplot  [thick, color = green, mark = none, mark size = 2, mark repeat = 0, mark phase = 0]
coordinates {
(100, 0.11227832749034335)(200, 0.04899507535426225)(300, 0.041319619707278425)(400, 0.027361049624437465)(500, 0.027731976246587246)
}; \addlegendentry{SS}
\addplot  [thick, color = violet, dashed , mark =none,  mark size = 2, mark repeat = 0, mark phase = 0]
coordinates {
(100, 0.09366561175627929)(200, 0.04597432617670799)(300, 0.03027900174564828)(400, 0.02647758124383778)(500, 0.017441996458604187)
}; \addlegendentry{PPI}
\addplot  [thick, color = violet,  mark = triangle,  mark size = 2, mark repeat = 0, mark phase = 0]
coordinates {
(100, 0.07196455334996854)(200, 0.03514447402656854)(300, 0.02434946702686502)(400, 0.02051716263406296)(500, 0.015206943710071537)
}; \addlegendentry{Tuned PPI}
\addplot  [thick, color =blue, dashed , mark = none,  mark size = 2, mark repeat = 0, mark phase = 0]
coordinates {
(100, 0.04647781287049195)(200, 0.02504665641589281)(300, 0.014646878126809926)(400, 0.011818131202139658)(500, 0.009214057349513326)
}; \addlegendentry{CPPI}
\addplot  [thick, color = blue, mark = square, mark size = 2, mark repeat = 0, mark phase = 0]
coordinates {
(100, 0.03767071955282291)(200, 0.02079378114242821)(300, 0.012253068330832893)(400, 0.010047615173936193)(500, 0.008129060743781013)
}; \addlegendentry{Tuned CPPI}
\end{axis}
\end{tikzpicture}}
\subfigure[$R^2=0.4$]  
{  
\begin{tikzpicture}[scale=0.6]
\begin{axis}[
tick align=outside,
tick pos=left,
x grid style={white!69.0196078431373!black},
xlabel={Size of the labeled dataset $n$},
xmajorgrids,
xmin= 100, xmax=400,
xtick style={color=black},
y grid style={white!69.0196078431373!black},
ylabel={Mean squared error},
ymajorgrids,
ymin=0.009, ymax=0.04,
ytick style={color=black},
grid=major,
scaled ticks = true,
legend pos = north east,
legend style={nodes={scale=0.6, transform shape}},
grid style=densely dashed,
]

\addplot  [thick, color = darkgray, mark = pentagon*, mark size = 2, mark repeat = 0, mark phase = 0]
coordinates {
(100, 0.038348589406394615)(200, 0.01994742299135502)(300, 0.012712656088921534)(400, 0.010970719147244698)(500, 0.009145579729673546)
}; \addlegendentry{ERM}
\addplot  [thick, color = green, mark = none, mark size = 2, mark repeat = 0, mark phase = 0]
coordinates {
(100, 0.09087088094799976)(200, 0.03861575891398912)(300, 0.027644116362673762)(400, 0.017901437936340494)(500, 0.01562645689642788)
}; \addlegendentry{SS}
\addplot  [thick, color = violet, dashed , mark =none,  mark size = 2, mark repeat = 0, mark phase = 0]
coordinates {
(100, 0.07398776829010201)(200, 0.03669532190225189)(300, 0.025303482439750054)(400, 0.021763900044000993)(500, 0.013866674156056554)
}; \addlegendentry{PPI}
\addplot  [thick, color = violet,  mark = triangle,  mark size = 2, mark repeat = 0, mark phase = 0]
coordinates {
(100, 0.06330738197738899)(200, 0.031492880018611086)(300, 0.02198808994960744)(400, 0.01839061131910692)(500, 0.012786317325964606)
}; \addlegendentry{Tuned PPI}
\addplot  [thick, color =blue, dashed , mark = none,  mark size = 2, mark repeat = 0, mark phase = 0]
coordinates {
(100, 0.03775877183499559)(200, 0.019892473495182358)(300, 0.012148130989899485)(400, 0.009724586383736229)(500, 0.007450499329942499)
}; \addlegendentry{CPPI}
\addplot  [thick, color = blue, mark = square, mark size = 2, mark repeat = 0, mark phase = 0]
coordinates {
(100, 0.03402773588913957)(200, 0.018103817717367184)(300, 0.01098697150934958)(400, 0.00884229242286316)(500, 0.006822745227370164)
}; \addlegendentry{Tuned CPPI}
\end{axis}
\end{tikzpicture}}
\subfigure[$R^2=0.75$]  
{  
\begin{tikzpicture}[scale=0.6]
\begin{axis}[
tick align=outside,
tick pos=left,
x grid style={white!69.0196078431373!black},
xlabel={Size of the labeled dataset $n$},
xmajorgrids,
xmin= 100, xmax=400,
xtick style={color=black},
y grid style={white!69.0196078431373!black},
ylabel={Mean squared error},
ymajorgrids,
ymin=0.002, ymax=0.04,
ytick style={color=black},
grid=major,
scaled ticks = true,
legend pos = north east,
legend style={nodes={scale=0.6, transform shape}},
grid style=densely dashed,
]

\addplot  [thick, color = darkgray, mark = pentagon*, mark size = 2, mark repeat = 0, mark phase = 0]
coordinates {
(100, 0.038205224815723984)(200, 0.020986112356226726)(300, 0.013077945904175582)(400, 0.010833219863844966)(500, 0.008576278805513931)
}; \addlegendentry{ERM}
\addplot  [thick, color = green, mark = none, mark size = 2, mark repeat = 0, mark phase = 0]
coordinates {
(100, 0.043063026507320475)(200, 0.01867160594731212)(300, 0.016836279184926715)(400, 0.012583238876945872)(500, 0.013387132767091978)
}; \addlegendentry{SS}
\addplot  [thick, color = violet, dashed , mark =none,  mark size = 2, mark repeat = 0, mark phase = 0]
coordinates {
(100, 0.03501720108428116)(200, 0.017457466183257792)(300, 0.011118050016002001)(400, 0.009595020170720503)(500, 0.006194802313007229)
}; \addlegendentry{PPI}
\addplot  [thick, color = violet,  mark = triangle,  mark size = 2, mark repeat = 0, mark phase = 0]
coordinates {
(100, 0.03714404822713604)(200, 0.01709130445844109)(300, 0.011230298964912224)(400, 0.009309577337728173)(500, 0.006125177820387484)
}; \addlegendentry{Tuned PPI}
\addplot  [thick, color =blue, dashed , mark = none,  mark size = 2, mark repeat = 0, mark phase = 0]
coordinates {
(100, 0.019128225789988858)(200, 0.0096201820708267)(300, 0.005551803406274322)(400, 0.004305942199718821)(500, 0.0031781539483262114)
}; \addlegendentry{CPPI}
\addplot  [thick, color = blue, mark = square, mark size = 2, mark repeat = 0, mark phase = 0]
coordinates {
(100, 0.020051261508457204)(200, 0.01005083952308111)(300, 0.005753272946929084)(400, 0.00431358875782876)(500, 0.0031619476078787413)
}; \addlegendentry{Tuned CPPI}
\end{axis}
\end{tikzpicture}}
\end{center}
\caption{Mean squared error as a function of the size of the labeled dataset for the problem of men estimation under the synthetic data-generation model \eqref{data_gen_model}. The results are averaged over $300$ trials.}
\label{mean_est}
\end{figure*}
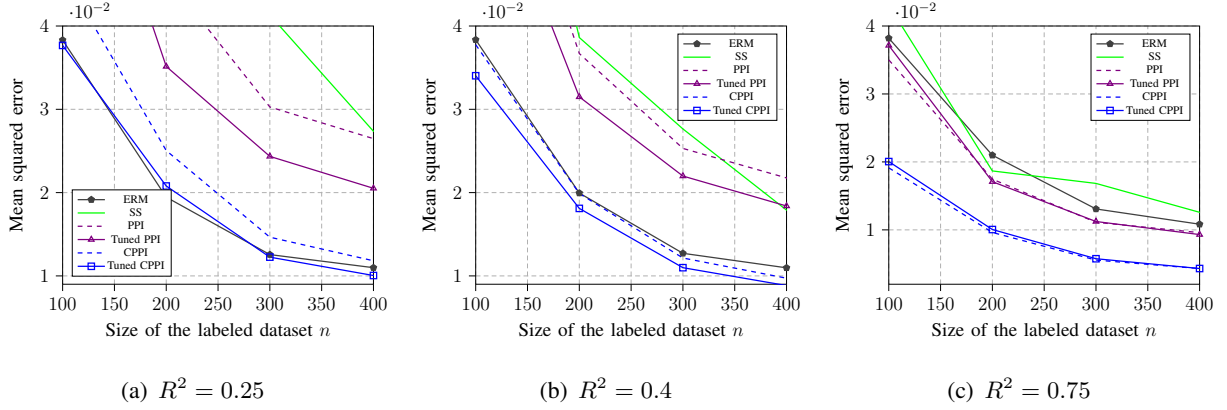

Given the data generation model in \eqref{data_gen_model}, the target here is the mean of variable $Y$, i.e., $\theta^\star=\ex[Y]=\mu$. Here, we fix $d=2$, $\sigma = 2$, and $\mu = 4$. \figref{mean_est} shows the performance of all schemes versus the number of data points in the labeled dataset, $n$, for different values of the correlation parameter $R$. The performance is measured in terms of the MSE between the true mean $\mu$ and its estimate. The proposed tuned CPPI guarantees the best performance under all settings. Specifically, when $R$ is small (\figref{mean_est}(a)), here $R^2=0.25$, which implies a low dependence of $Y$ on the feature vector $X$, tuned CPPI performs close to ERM, while outperforming CPPI, as the trained models are not expected to help in this case. Conversely, for high values of $R$ (\figref{mean_est}(c)), here $R^2=0.75$, tuned CPPI and CPPI have comparable performance and significantly outperform ERM. In the intermediate regime, corresponding to  $R^2=0.4$, tuned CPPI outperforms both CPPI and ERM (\figref{mean_est}(b)).

\subsection{Linear Regression}
Using the data generation model \eqref{data_gen_model} with $d=3$, $\mu = 0$, and $\sigma = 2$, we consider here estimating the coefficients of a linear regression model. following a standard feature selection methodology~\cite{BarrieWetherill1986}, only the first two features $(X_1, X_2)$ are included as covariates, that is, the target is the parameter vector $\theta^\star\in \mathbb{R}^2$ given by
\begin{align}
\theta^\star = \argmin_\theta \ex \left[(Y - X_{\rm red}^T\theta)^2\right],
\end{align}
where $X_{\rm red}= [X_1,X_2]^T$. The performance metric is the MSE between the true parameter vector $\theta^\star = [\beta_1,\beta_2]^T$ and its estimate.

We report in \figref{LR_est} the performance of all schemes versus the labeled dataset size $n$ for different values of $R$. We observe a similar behavior as in the mean estimation case. In particular, tuned CPPI guarantees the best performance in all settings, yielding better results than CPPI for low values of $R$, outperforming ERM for high values of $R$, and improving over both schemes in the intermediate regime of parameter $R$. Moreover, we note that the SS scheme, which disregards the prediction bias of the trained models, demonstrates significantly inferior performance.

\begin{figure*}[t]
\begin{center}

\subfigure[$R^2=0.1$]  
{  
\begin{tikzpicture}[scale=0.57]
\begin{axis}[
tick align=outside,
tick pos=left,
x grid style={white!69.0196078431373!black},
xlabel={Size of the labeled dataset $n$},
xmajorgrids,
xmin= 100, xmax=400,
xtick style={color=black},
y grid style={white!69.0196078431373!black},
ylabel={Mean squared error},
ymajorgrids,
ymin=0.01, ymax=0.15,
ytick style={color=black},
grid=major,
scaled y ticks=false,
yticklabel=\pgfkeys{/pgf/number format/.cd,fixed,precision=2,zerofill}\pgfmathprintnumber{\tick},
% legend pos = north east,
legend style={nodes={scale=0.6, transform shape}},
grid style=densely dashed,
]

\addplot  [semithick, color = black, mark = star, mark size = 2, mark repeat = 0, mark phase = 0]
coordinates {
(100, 0.070675690218214)(200, 0.039601915375337145)(300, 0.028489807904329697)(400, 0.020612161935576818)(500, 0.017829449179455434)
}; \addlegendentry{ERM}
\addplot  [thick, color = green, mark = none, mark size = 2, mark repeat = 0, mark phase = 0]
coordinates {
(100, 0.17710397777977557)(200, 0.10585892043567052)(300, 0.06594327162582075)(400, 0.05784348399401768)(500, 0.04555204757378674)
}; \addlegendentry{SS}

\addplot  [thick, color = violet, dashed , mark =none,  mark size = 2, mark repeat = 0, mark phase = 0]
coordinates {
(100, 0.16392652134125557)(200, 0.09075796107190237)(300, 0.06146592139769022)(400, 0.04267047742163092)(500, 0.034290335900861696)
}; \addlegendentry{PPI}

\addplot  [thick, color = violet,  mark = triangle,  mark size = 2, mark repeat = 0, mark phase = 0]
coordinates {
(100, 0.14895109986517116)(200, 0.07972174113984612)(300, 0.056101378817890256)(400, 0.037661891543907436)(500, 0.031714267681520705)
}; \addlegendentry{Tuned PPI}

\addplot  [thick, color =blue, dashed , mark = none,  mark size = 2, mark repeat = 0, mark phase = 0]
coordinates {
(100, 0.0846937861710592)(200, 0.04705446891336431)(300, 0.032064059778928365)(400, 0.02302975972092665)(500, 0.018388711146680805)
 }; \addlegendentry{CPPI}
\addplot  [thick, color = blue, mark = square, mark size = 2, mark repeat = 0, mark phase = 0]
coordinates {
(100, 0.07111684231852872)(200, 0.0393803486161813)(300, 0.028662766051141285)(400, 0.02076164287973274)(500, 0.017011445290394615)
}; \addlegendentry{Tuned CPPI}

\end{axis}
\end{tikzpicture}}
\subfigure[$R^2=0.4$]  
{  
\begin{tikzpicture}[scale=0.57]
\begin{axis}[
tick align=outside,
tick pos=left,
x grid style={white!69.0196078431373!black},
xlabel={Size of the labeled dataset $n$},
xmajorgrids,
xmin= 100, xmax=400,
xtick style={color=black},
y grid style={white!69.0196078431373!black},
ylabel={Mean squared error},
ymajorgrids,
ymin=0.01, ymax=0.15,
ytick style={color=black},
grid=major,
scaled y ticks=false,
legend style={at={(0.65,1.0)}},
yticklabel=\pgfkeys{/pgf/number format/.cd,fixed,precision=2,zerofill}\pgfmathprintnumber{\tick},
% legend pos = north east,
legend style={nodes={scale=0.6, transform shape}},
grid style=densely dashed,
]
\addplot  [semithick, color = black, mark = star, mark size = 2, mark repeat = 0, mark phase = 0]
coordinates {
(100, 0.07133436221630608)(200, 0.040061778967507955)(300, 0.026924169651654263)(400, 0.01996031563371543)(500, 0.018444861990462184)
}; \addlegendentry{ERM}
\addplot  [thick, color = green, mark = none, mark size = 2, mark repeat = 0, mark phase = 0]
coordinates {
(100, 0.43651505371478977)(200, 0.23580224771682376)(300, 0.14625719263827022)(400, 0.11455490504722278)(500, 0.0921973785855292)
}; \addlegendentry{SS}
\addplot  [thick, color = violet, dashed , mark =none,  mark size = 2, mark repeat = 0, mark phase = 0]
coordinates {
(100, 0.15437717476473672)(200, 0.07780486709852412)(300, 0.052602353791441764)(400, 0.03576574776784284)(500, 0.026346154673187502)
}; \addlegendentry{PPI}
\addplot  [thick, color = violet,  mark = triangle,  mark size = 2, mark repeat = 0, mark phase = 0]
coordinates {
(100, 0.13695120773588096)(200, 0.07004498056327527)(300, 0.04839171278396889)(400, 0.0328315891014417)(500, 0.025934031856541716)
}; \addlegendentry{Tuned PPI}
\addplot  [thick, color =blue, dashed , mark = none,  mark size = 2, mark repeat = 0, mark phase = 0]
coordinates {
(100, 0.07503007913961231)(200, 0.039192497353618926)(300, 0.025153073039370265)(400, 0.01804753981538715)(500, 0.01448936029630366)
 }; \addlegendentry{CPPI}
\addplot  [thick, color = blue, mark = square, mark size = 2, mark repeat = 0, mark phase = 0]
coordinates {
(100, 0.06254122922130405)(200, 0.03460249425638817)(300, 0.023154114513263352)(400, 0.01688593613940827)(500, 0.013823208789439694)
}; \addlegendentry{Tuned CPPI}
\end{axis}
\end{tikzpicture}}
\subfigure[$R^2=0.75$]  
{ 
\begin{tikzpicture}[scale=0.57]
\begin{axis}[
tick align=outside,
tick pos=left,
x grid style={white!69.0196078431373!black},
xlabel={Size of the labeled dataset $n$},
xmajorgrids,
xmin= 100, xmax=400,
xtick style={color=black},
y grid style={white!69.0196078431373!black},
ylabel={Mean squared error},
ymajorgrids,
ymin=0, ymax=0.15,
ytick style={color=black},
grid=major,
legend style={nodes={scale=0.6, transform shape}},
grid style=densely dashed,
legend style={at={(0.65,1.0)}},
scaled y ticks=false,
yticklabel=\pgfkeys{/pgf/number format/.cd,fixed,precision=2,zerofill}\pgfmathprintnumber{\tick},
]
\addplot  [semithick, color = black, mark = star, mark size = 2, mark repeat = 0, mark phase = 0]
coordinates {
(100, 0.07598316617086147)(200, 0.04121814771307011)(300, 0.025201952832024043)(400, 0.01907815528427858)(500, 0.01854890996003476)
}; \addlegendentry{ERM}
\addplot  [thick, color = green, mark = none, mark size = 2, mark repeat = 0, mark phase = 0]
coordinates {
(100, 0.41372273593992115)(200, 0.2409682200673317)(300, 0.1577173219253055)(400, 0.1260356577857746)(500, 0.10258721352445085)
}; \addlegendentry{SS}

\addplot  [thick, color = violet, dashed , mark =none,  mark size = 2, mark repeat = 0, mark phase = 0]
coordinates {
(100, 0.16011280251402107)(200, 0.07096153972326948)(300, 0.04119132214273373)(400, 0.027559194223723733)(500, 0.019197178839239173)
}; \addlegendentry{PPI}

\addplot  [thick, color = violet,  mark = triangle,  mark size = 2, mark repeat = 0, mark phase = 0]
coordinates {
(100, 0.12933224910029786)(200, 0.06251688738780474)(300, 0.03709249119450911)(400, 0.0250621969769986)(500, 0.01859212632013628)
}; \addlegendentry{Tuned PPI}

\addplot  [thick, color =blue, dashed , mark = none,  mark size = 2, mark repeat = 0, mark phase = 0]
coordinates {
(100, 0.0699504154561443)(200, 0.0323330370562257)(300, 0.018632044503738256)(400, 0.013124516400655619)(500, 0.010545839700501153)
 }; \addlegendentry{CPPI}
\addplot  [thick, color = blue, mark = square, mark size = 2, mark repeat = 0, mark phase = 0]
coordinates {
(100, 0.05612735664979723)(200, 0.02816380515900052)(300, 0.01620242000140961)(400, 0.01165778942889563)(500, 0.009763453148904554)
}; \addlegendentry{Tuned CPPI}

\end{axis}
\end{tikzpicture}}
\end{center}
\caption{Mean squared error as a function of the size of the labeled dataset for the problem of linear regression under the synthetic data-generation model \eqref{data_gen_model}. The results are averaged over $300$ trials.}
\label{LR_est}
\end{figure*}
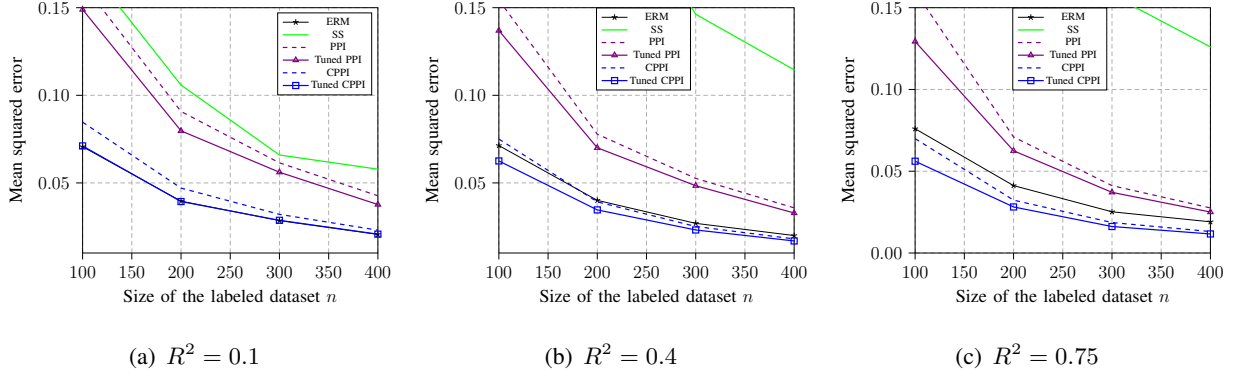

\section{Prediction-powered RSSI localization}
\label{loc_application}
In this section, we consider the problem of indoor localization based on received signal strength information (RSSI)~\cite{8371230}. RSSI-based localization is a popular positioning technique~\cite{5290385,6244790}, in which RSS measurements from multiple access points (APs) are used to infer the location of a user equipment (UE). We investigate the scenario in which a labeled dataset containing RSSI measurements $X$ and the corresponding UE locations $Y$ is available, along with an unlabeled dataset with only RSSI measurements $\tilde X$. For instance, the unlabeled dataset could contain RSS measurements collected by volunteers who do not reveal their location information~\cite{li2023exploiting}.

\subsection{System Model}
We consider an indoor environment $m$ access points (APs)~\cite{nabati2020using}. Given a dataset $\mathcal{D} = \{(X_i, Y_i)\}_{i=1}^n$ of RSSI measurements $X_i\in\mathbb{R}^m$ and the corresponding location information (longitude and latitude) $Y_i \in \mathbb{R}^2$, as well as an unlabeled dataset consisting of RSSI measurements $\tilde X_i\in\mathbb{R}^m$. 

As in \cite{yan2021extreme}, the goal is to train an extreme learning machine (ELM) regression model to predict the label $Y$ given RSSI measurements $X$~\cite{huang2006extreme,huang2011extreme}. An ELM model is a single hidden layer neural network, in which the weights connecting the hidden layer and the output layer are optimized, while other parameters are randomly initialized and kept fixed. With $p$ neurons in the hidden layer, define as $W\in\mathbb{R}^{p\times m}$ the weight matrix between the input and the hidden layer, and as $b\in\mathbb{R}^{p}$ the bias vector at the hidden layer. The ELM model predicts the label as~\cite{huang2011extreme}
\begin{align}
    \hat Y =  \left[h(X)^T \theta_1, h(X)^T \theta_2\right]^T,
\end{align}
where $\sigma(\cdot)$ is an activation function; $\theta = [\theta_1^T, \theta_2^T]^T \in \mathbb{R}^{2p}$ is a parameter vector to be optimized; and $h(X) =\sigma(W X + b)$. The loss function is given as~\cite{huang2011extreme}
\begin{align}
\ell_\theta(X,Y) = \sum_{j=1}^2\ex \left[ (Y[j]-h(X)^T\theta_j)^2\right] + \gamma \sum_{j=1}^2\norm{\theta_j}^2_2,
\end{align}
with $Y[j]$ being the $j$-th entry of vector $Y$. 

In order to reduce the data requirements for training the ELM model, we leverage the availability of unlabeled data $\tilde X$ by using ML prediction models. Specifically, the PPI model $f(X)$ and the CPPI models $\{f^{(k)}(X)\}$ are implemented here as a fully connected neural network with three hidden layers of size 256, 128, and 32, respectively, with a LeakyRelu activation function. 

% Under this setting, a subset of the feature vector is missing in the unlabeled data compared to the labeled one. As a result, the predictive models $\{f^{(k)}\}$ for CPPI and tuned CPPI are trained to predict the missing feature information given the known features and the labels. 

% More explicitly, we assume in our experiments that for the unlabeled data, the RSS measurements are available from the first $m_1$ APs denoted as $\tilde X = [x_{1},\cdots ,x_{m_1}]^T$ with $x_i$ being the RSS from AP $i$.  

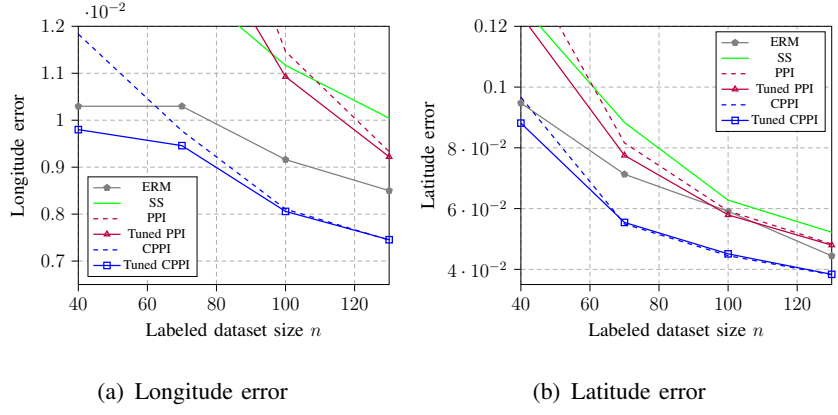
\begin{figure}[t]
\begin{center}
\subfigure[Longitude error]  
{  
\begin{tikzpicture}[scale=0.6]
\begin{axis}[
tick align=outside,
tick pos=left,
x grid style={white!69.0196078431373!black},
xlabel={Labeled dataset size $n$},
ylabel={Longitude error},
xmajorgrids,
xmin= 40, xmax=130,
ymin=0.0065, ymax=0.012,
xtick style={color=black},
y grid style={white!69.0196078431373!black},
ymajorgrids,
ytick style={color=black},
grid=major,
scaled ticks = true,
legend pos = south west,
legend style={nodes={scale=0.7, transform shape}},
grid style=densely dashed,
]
\addplot  [thick, color = gray, mark = pentagon*, mark size = 2, mark repeat = 0, mark phase = 0]
coordinates {
(40, 0.010297337973106497)(70, 0.01029697525716408)(100, 0.009161283371945897)(130, 0.008500203063557447)(160, 0.008323487110476428)(190, 0.008078509939076212)(220, 0.007561824022237866)
}; \addlegendentry{ERM}
\addplot  [thick, color = green, mark = none, mark size = 2, mark repeat = 0, mark phase = 0]
coordinates {
(40, 0.018632011758214734)(70, 0.01302914554333463)(100, 0.01116820878501686)(130, 0.010043652216648854)(160, 0.008784125739800258)(190, 0.00858025381251872)(220, 0.00824439190452909)
}; \addlegendentry{SS}
\addplot  [thick, color = purple, dashed , mark =none,  mark size = 2, mark repeat = 0, mark phase = 0]
coordinates {
(40, 0.03115294710282869)(70, 0.017309179831415524)(100, 0.011464783857477455)(130, 0.009332148165847131)(160, 0.008801297430852228)(190, 0.008266206310928455)(220, 0.007603925684059627)
}; \addlegendentry{PPI}
\addplot  [thick, color = purple,  mark = triangle,  mark size = 2, mark repeat = 0, mark phase = 0]
coordinates {
(40, 0.022034065314194547)(70, 0.014642298018612977)(100, 0.010927085315882412)(130, 0.009220450247784928)(160, 0.008695819094572838)(190, 0.008196885576235454)(220, 0.007570954706834122)
}; \addlegendentry{Tuned PPI}
\addplot  [thick, color =blue, dashed , mark = none,  mark size = 2, mark repeat = 0, mark phase = 0]
coordinates {
(40, 0.011833502054935297)(70, 0.009773975169333749)(100, 0.008105956919593447)(130, 0.007447241993182553)(160, 0.00765786349946993)(190, 0.007278177374363486)(220, 0.00688709793405537)
}; \addlegendentry{CPPI}
\addplot  [thick, color = blue, mark = square, mark size = 2, mark repeat = 0, mark phase = 0]
coordinates {
(40, 0.00979934406450369)(70, 0.009458314725676802)(100, 0.008057984633302564)(130, 0.007454772668342289)(160, 0.00766616979658576)(190, 0.00730277926074455)(220, 0.0069023663309773255)
}; \addlegendentry{Tuned CPPI}
\end{axis}
\end{tikzpicture}}
\subfigure[Latitude error]  
{ 
\begin{tikzpicture}[scale=0.6]
\begin{axis}[
tick align=outside,
tick pos=left,
x grid style={white!69.0196078431373!black},
xlabel={Labeled dataset size $n$},
ylabel={Latitude error},
xmajorgrids,
xmin= 40, xmax=130,
ymin=0.035, ymax=0.12,
xtick style={color=black},
y grid style={white!69.0196078431373!black},
ymajorgrids,
ytick style={color=black},
grid=major,
scaled ticks = true,
% legend pos = south east,
legend style={nodes={scale=0.7, transform shape}},
grid style=densely dashed,
]
\addplot  [thick, color = gray, mark = pentagon*, mark size = 2, mark repeat = 0, mark phase = 0]
coordinates {
(40, 0.09475864363300751)(70, 0.071295017000384)(100, 0.05909793822058514)(130, 0.04444559376668023)(160, 0.042495851248909)(190, 0.040073381258375014)(220, 0.041637536667877055)
}; \addlegendentry{ERM}
\addplot  [thick, color = green, mark = none, mark size = 2, mark repeat = 0, mark phase = 0]
coordinates {
(40, 0.12713455024564962)(70, 0.08835293496306607)(100, 0.06278863069704127)(130, 0.05221602656497257)(160, 0.046181202358014956)(190, 0.04210844333308452)(220, 0.0426284796220007)
}; \addlegendentry{SS}
\addplot  [thick, color = purple, dashed , mark =none,  mark size = 2, mark repeat = 0, mark phase = 0]
coordinates {
(40, 0.14273166097574477)(70, 0.08161676301393371)(100, 0.05924756513228904)(130, 0.048326273126386085)(160, 0.0447422904607648)(190, 0.04061404491306132)(220, 0.04077921005499952)
}; \addlegendentry{PPI}
\addplot  [thick, color = purple,  mark = triangle,  mark size = 2, mark repeat = 0, mark phase = 0]
coordinates {
(40, 0.12360762332516381)(70, 0.0775232984275095)(100, 0.05791818584090778)(130, 0.04798169739061884)(160, 0.04434329560684611)(190, 0.04043572847178569)(220, 0.040671237125945406)
}; \addlegendentry{Tuned PPI}
\addplot  [thick, color =blue, dashed , mark = none,  mark size = 2, mark repeat = 0, mark phase = 0]
coordinates {
(40, 0.09667481952930844)(70, 0.0549489312251408)(100, 0.044502755137861796)(130, 0.03819758769837258)(160, 0.03972912733420482)(190, 0.035159333710084936)(220, 0.03757616070251002)
}; \addlegendentry{CPPI}
\addplot  [thick, color = blue, mark = square, mark size = 2, mark repeat = 0, mark phase = 0]
coordinates {
(40, 0.08813999097858767)(70, 0.05543871976949692)(100, 0.04510012136547817)(130, 0.03835919656529577)(160, 0.039704097897016115)(190, 0.035358634758717034)(220, 0.03767524599078778)
}; \addlegendentry{Tuned CPPI}
\end{axis}
\end{tikzpicture}}
\end{center}
\caption{Localization error as a function labeled dataset size $n$. the results are averaged over $100$ trials.}
\label{localization}
\end{figure}

\subsection{Numerical Results}
\label{num_res}
We compare the performance of tuned CPPI with the benchmark schemes using a real dataset from \cite{torres2014ujiindoorloc}, which contains indoor RSS measurements from multiple WiFi APs with the corresponding longitude and latitude information in several floors and three buildings. The dataset can be used for both regression, predicting longitude and latitude, and classification, predicting floor number and building. Since our focus here is regression, we consider in our experiments a subset of the data corresponding to one floor in a specific building. The considered subset contains 2097 samples, from which $20\%$ of samples are reserved for testing. Given a labeled data size $n$, the training data is randomly divided into a labeled dataset of size $n$ and an unlabeled dataset of size $N$. For PPI and tuned PPI, the labeled dataset is divided into two equal parts for fitting model $f(X)$ and for computing the rectifier term.

\figref{localization} reports longitude and latitude errors of all schemes as a function of the labeled dataset size $n$. We note that the tuned CPPI scheme guarantees the best performance and outperforms CPPI for smaller values of $n$. Moreover, both CPPI and tuned CPPI provide a significant gain as compared to ERM, while PPI, tuned PPI, and conventional SS (with $\gamma=1$) perform worse than ERM.

\section{Prediction-powered Beam Alignment in mmWave Massive MIMO}
\label{beam_alignment}
Beamforming design is a crucial task in mmWave massive MIMO systems~\cite{giordani2018tutorial,yang2023hierarchical,zhang2020adaptive}, as beamforming is necessary to compensate for the more severe path loss experienced at higher carrier frequencies~\cite{kutty2015beamforming}. Codebook-based beam alignment consists of selecting the best beam from codebooks of predefined beams based on beam sweeping, which requires transmission of pilot signals~\cite{kutty2015beamforming,yang2023hierarchical}.
% Thanks to its larger bandwidth, mmWave communication can offer orders of magnitude higher data rates than classical sub-6 GHz communication systems. However, mmWave signals undergo more severe propagation environments resulting in higher path loss~\cite{kutty2015beamforming}. Fortunately, with mmWave short wavelength, more antennas can be packed in a small physical area allowing more precise beamforming, which 
% 

Recently, a new approach has emerged that alleviates the training overhead by leveraging the concept of a channel knowledge map (CKM). As discussed in Sec.~\ref{Intro}, a CKM is a site-specific database of channel information linked to transmitter and receiver locations~\cite{zeng2021toward,wu2021environment}.
% This database simplifies the process of acquiring real-time channel state information (CSI), potentially eliminating the need for explicit training.
In this section, we propose a method that trains a mapping between a user's location and a pair of beams within their respective codebooks by leveraging both labeled and unlabeled data. As detailed below and illustrated in \figref{app_beam}, labeled data consists of a user's location and the corresponding CSI, while unlabeled data only includes a user's location. Furthermore, for the unlabeled inputs $X$, CSI parameters are estimated using a CKM.

\subsection{System Model}
\label{sys_model}
As in~\cite{zeng2021toward,wu2021environment}, we consider a downlink mmWave massive MIMO communication system, in which a base station (BS) equipped with $N^{TX}$ transmit antennas communicates with a user equipment (UE) with $N^{RX}$
receive antennas. We assume that both BS and UE have a single radio frequency chain, and that beamforming is achieved using phase shifters in the analog domain. The beamforming vectors at the BS and the UE, denoted as $u\in\mathbb{C}^{N^{TX}\times 1} $ and $w\in\mathbb{C}^{N^{RX}\times 1} $, respectively, are ideally chosen from finite codebooks $\mathcal{U}$ and $\mathcal{W}$ by maximizing the signal-to-noise ratio (SNR). Indexing the set of all beamforming pairs $(u,w) \in \mathcal{U}\times \mathcal{W}$ as $(u_j,w_j)$ with $j\in \mathcal{J} = \{1,\cdots,J = |\mathcal{U}||\mathcal{W}|\} $, the optimal beam index is obtained as
\begin{align}
    Y = \argmax_{j \in \mathcal{J}} \ \left|u_j^H Hw_j\right|^2,
    \label{beam_selection}
\end{align}
where $H \in \mathbb{C}^{N^{TX} \times N^{RX}}$ is the current channel matrix. Accordingly, the ideal beam selection in \eqref{beam_selection} requires knowledge of the current channel matrix $H$. 

When the current CSI is unknown, an alternative is to select the beams $(u_j,w_j)$ based solely on the location $X \in \mathbb{R}^3$ of the UE~\cite{zeng2021toward}. To this end, we adopt here a softmax regression model that maps a feature vector $\psi(X)\in \mathbb{R}^m$ of the UE location $X$ to the corresponding estimated beam index $\hat Y\in \mathcal{J}$. The feature vector $\psi(X)$ is obtained here using the radial basis function (RBF) kernel with Nystr\"{o}m approximation \cite{NIPS2000_19de10ad}.

The beam index is then selected based on the location $X$ as
\begin{align}
\hat{Y} = \argmax_{j \in \mathcal{J}}  \left\{\frac{\exp(\theta_j^T \psi(X))}{\sum_{j'\in \mathcal{J}}\exp(\theta_{j'}^T \psi(X)}\right\},
\label{linear_softmax_reg}
\end{align}
where $\theta = [\theta_1^T, \cdots,  \theta_J^T]^T$, with $\theta_j\in\mathbb{R}^m$, is a parameter vector. The optimal beam selection parameter $\theta^\star$ is obtained as the solution of the problem
\begin{align}
   \theta^\star = \argmin_\theta \ \ex [\ell_\theta(X, Y)],
    \label{opt_param}
\end{align}
where the average is over the joint distribution of the UE location $X$ and of the corresponding optimal beam $Y$ in \eqref{beam_selection}; and the loss function $\ell_\theta(X, Y)$ is the regularized softmax regression loss~\cite{menard2002applied}
\begin{align}
\ell_\theta(X, Y)= - \sum_{j\in\mathcal{J}}\mathbbm{1}_{\{Y=j\}}\log \frac{\exp(\theta_j^T \psi(X))}{\sum\limits_{j'\in \mathcal{J}}^{}\exp(\theta_{j'}^T \psi(X))} + \gamma  \norm{\theta}^2_2,
\label{loss_func}
\end{align}
with $\gamma\geq 0$ being a regularization parameter. Problem \eqref{opt_param} is in the form \eqref{estimand}, since the loss $\ell_\theta(X, Y)$ is a convex function of parameters $\theta$.

% One can use any other classification model that can be written in the form of \eqref{estimand}. We have chosen softmax regression with kernel-based feature extraction as an example to illustrate the tuned CPPI approach since this model yields competitive performance on this problem.
To address problem \eqref{opt_param}, we assume access to a labeled dataset $\mathcal{D} = \{( X_i, Y_i)\}_{i=1}^n$ where label $Y$ is obtained from available CSI using \eqref{beam_selection}, as well as an unlabeled dataset $\tilde{\mathcal{D}} = \{\tilde X_i\}_{i=1}^N$ drawn from the distribution $P_X$ of the UE's locations. Note that prior information about the UE location distribution can be more easily acquired than CSI via passive tracking systems for instance~\cite{ribeiro2022passive,redondi2018building}.

% Next, we explore how the accuracy of the softmax regression model in \eqref{opt_param} can be enhanced, using our proposed tuned CPPI approach, in the setting where both labeled and unlabeled datasets are available. The unlabeled dataset $\tilde{\mathcal{D}} = \{\tilde X_i\}_{i=1}^N$ contains new UE location samples $\tilde X_i$, which can be assumed to be generated by the digital twin (DT) with some prior information about the UE location distribution. Such information can be easily acquired using data from the physical twin (PT). 

\subsection{CKM-Based Beam Alignment}

In order to use the unlabeled data of the user's locations, it is necessary to train a model $f(X)$ that maps UE location $X$ to beam index $Y$. To this end, we consider a model of the form
\begin{align}
    Y = f(X) = g(c(X)),
    \label{model_f_beam}
\end{align}
where $c(X)$ is the CKM-based mapping between position $X$ and CSI matrix $H$, while function $g(\cdot)$ maps CSI $H$ to the optimal beam index by solving \eqref{beam_selection} with $H=c(X)$. 

To define the CKM-based mapping function $c(X)$, we follow \cite{wu2021environment}, which assigns to each position $X$ multi-path information $Z$ defined as
% A CKM is a general concept that encompasses many types of channel information. For instance, CKM can be a channel matrix map (CMM) mapping a UE position $X$ to the corresponding channel matrix $H$~\cite{wu2021environment}. From a practical point of view, a CMM requires excessive storage and computation resources. An alternative to CMM with lower complexity requirements is a channel path map (CPM), where only information about the channel path (such as the number of significant paths and their power, phases, and AoAs/AoDs) is stored \cite{wu2021environment}. Specifically, CPM is a map between the UE location $X$ and the corresponding path information $Z$ defined as
\begin{align}
Z = \left\{ L, \{ \alpha^l, \theta_{AoD}^l, \phi_{AoD}^l,  \theta_{AoA}^l, \phi_{AoA}^l\}_{l=1}^L \right\},
\label{Z_expression}
\end{align}
where $L$ is the number of significant paths; $\alpha^l$ is the complex gain of the $l$-th path; $\theta_{AoA}^l$ and $\phi_{AoA}^l$ are the zenith and azimuth angles of arrival (AoA); and $\theta_{AoD}^l$ and $\phi_{AoD}^l$ are the zenith and azimuth angles of departure (AOD). From the path information $Z$, one can construct an estimate of the channel matrix $H$ using a standard multipath model
\begin{align}
  H = h(Z).
  \label{channel_func}
\end{align}

In order to enable the training of the CKM function $H=c(X)$, and thus of the model $f(X) $ in \eqref{model_f_beam}, we assume access to a labeled dataset $\mathcal{D}_Z = \{X_i,Z_i\}_{i=1}^n$ containing pairs of UE locations $X_i$ and the corresponding path information $Z_i$. Using the channel function \eqref{channel_func} and the optimal beam index \eqref{beam_selection}, one can recover the labeled data $\mathcal{D} = \{X_i,Y_i\}_{i=1}^n$.

\begin{figure}[ht]
\begin{center}
\subfigure[$n = 300$]  
{ 
\begin{tikzpicture}[scale=0.6]
\begin{axis}[
tick align=outside,
tick pos=left,
x grid style={white!69.0196078431373!black},
xlabel={Number of BS antennas},
xmajorgrids,
xmin= 50, xmax=300,
xtick style={color=black},
y grid style={white!69.0196078431373!black},
ylabel={Channel capacity (bps/Hz)},
ymajorgrids,
ymin=2.2, ymax=6.2,
ytick style={color=black},
grid=major,
scaled ticks = true,
legend pos = north west,
legend style={nodes={scale=0.6, transform shape}},
grid style=densely dashed,
]

\addplot  [thick, color = darkgray, mark = pentagon*, mark size = 2, mark repeat = 0, mark phase = 0]
coordinates {
(50, 2.8462321504523986)(100, 2.8298740082887175)(200, 2.4747949110130256)(300, 2.4718935382231493)(400, 2.4589148074337412)
}; \addlegendentry{ERM}

\addplot  [thick, color = green, mark = none, mark size = 2, mark repeat = 0, mark phase = 0]
coordinates {
(50, 3.127171245766773)(100, 2.86681174369333)(200, 2.7120769123778374)(300, 2.480279706807535)(400, 2.2693412763279803)
}; \addlegendentry{SS}

\addplot  [thick, color = violet, dashed , mark =none,  mark size = 2, mark repeat = 0, mark phase = 0]
coordinates {
(50, 3.1091302917510526)(100, 3.1064420386328306)(200, 2.978899322031873)(300, 2.850668048910205)(400, 2.9396287616453227)
}; \addlegendentry{PPI}

\addplot  [thick, color = violet,  mark = triangle,  mark size = 2, mark repeat = 0, mark phase = 0]
coordinates {
(50, 3.2810355834994636)(100, 3.4360280793213196)(200, 3.2642934704791373)(300, 3.462602429364957)(400, 3.3148865618234407)
}; \addlegendentry{Tuned PPI}

\addplot  [thick, color =blue, dashed , mark = none,  mark size = 2, mark repeat = 0, mark phase = 0]
coordinates {
(50, 3.5897314285207305)(100, 3.71797192731885)(200, 3.7669154054205753)(300, 3.7941842493585143)(400, 3.7941553706952926)
}; \addlegendentry{CPPI}

\addplot  [thick, color = blue, mark = square, mark size = 2, mark repeat = 0, mark phase = 0]
coordinates {
(50, 3.708392834943368)(100, 3.942248364762279)(200, 4.04301454960707)(300, 4.0556940508024635)(400, 4.081637986368768)
}; \addlegendentry{Tuned CPPI}

\addplot  [semithick, color = black, mark = star, mark size = 2, mark repeat = 0, mark phase = 0]
coordinates {
(50, 4.426156413569589)(100, 5.042716906664163)(200, 5.700992597706176)(300, 6.113098842937517)(400, 6.410515691911261)
}; \addlegendentry{Perfect-CSI}

% \addplot  [semithick, color = red, mark = o, mark size = 2, mark repeat = 0, mark phase = 0]
% coordinates {
% (50, 3.71)(100, 3.89)(200, 4.03)(300, 4.10)(400, 4.204)
% }; \addlegendentry{Location-based}

\end{axis}
\end{tikzpicture}}
\subfigure[$n = 3000$]  
{  
\begin{tikzpicture}[scale=0.6]
\begin{axis}[
tick align=outside,
tick pos=left,
x grid style={white!69.0196078431373!black},
xlabel={Number of BS antennas},
xmajorgrids,
xmin= 50, xmax=300,
xtick style={color=black},
y grid style={white!69.0196078431373!black},
ylabel={Channel capacity (bps/Hz)},
ymajorgrids,
ymin=3.9, ymax=6.2,
ytick style={color=black},
grid=major,
scaled ticks = true,
legend pos = north west,
legend style={nodes={scale=0.6, transform shape}},
grid style=densely dashed,]
\addplot  [thick, color = darkgray, mark = pentagon*, mark size = 2, mark repeat = 0, mark phase = 0]
coordinates {
% (50, 4.013436658757879)(100, 4.258884191997132)(200, 4.325524170434595)(300, 4.501713339576551)(400, 4.083989091128629)
(50, 3.978325455875757)(100, 4.305812790488412)(200, 4.364690547432143)(300, 4.3088249225573945)(400, 4.135770834021695)
}; \addlegendentry{ERM}
\addplot  [thick, color = green, mark = none, mark size = 2, mark repeat = 0, mark phase = 0]
coordinates {
(50, 3.9400667384958776)(100, 4.289113227517361)(200, 4.387841194681612)(300, 4.36072378870601)(400, 4.162944316911362)
% (50, 3.546461975095726)(100, 3.6193328913214815)(200, 3.4783780432668454)(300, 3.367118464690729)(400, 3.321030334602361)
}; \addlegendentry{SS}
\addplot  [thick, color = violet, dashed , mark =none,  mark size = 2, mark repeat = 0, mark phase = 0]
coordinates {
(50, 3.973742070298371)(100, 4.343628025175115)(200, 4.662660785457252)(300, 4.760743391797428)(400, 4.885505322283333)
}; \addlegendentry{PPI}
\addplot  [thick, color = violet,  mark = triangle,  mark size = 2, mark repeat = 0, mark phase = 0]
coordinates {
% (50, 3.973742070298371)(100, 4.343628025175115)(200, 4.662660785457252)(300, 4.760743391797428)(400, 4.885505322283333)
(50, 3.9914237996745414)(100, 4.407377712403373)(200, 4.674837160560732)(300, 4.82866734756381)(400, 4.833400742724301)
}; \addlegendentry{Tuned PPI}
\addplot  [thick, color =blue, dashed , mark = none,  mark size = 2, mark repeat = 0, mark phase = 0]
coordinates {
% (50, 4.1578485719067935)(100, 4.674573626930951)(200, 5.029979052000006)(300, 5.179983143325941)(400, 5.074586005407174)
(50, 4.122242967020595)(100, 4.539280321626003)(200, 4.956654767902216)(300, 5.1147089709683895)(400, 5.113108902870125)
}; \addlegendentry{CPPI}

\addplot  [thick, color = blue, mark = square, mark size = 2, mark repeat = 0, mark phase = 0]
coordinates {
(50, 4.13787744434379)(100, 4.553947201798361)(200, 4.8993033686150556)(300, 5.1283061983318685)(400, 5.109038041703121)
% (50, 4.13787744434379)(100, 4.553947201798361)(200, 4.8993033686150556)(300, 5.1283061983318685)(400, 5.109038041703121)
}; \addlegendentry{Tuned CPPI}
\addplot  [semithick, color = black, mark = star, mark size = 2, mark repeat = 0, mark phase = 0]
coordinates {
(50, 4.426156413569589)(100, 5.042716906664163)(200, 5.700992597706176)(300, 6.113098842937517)(400, 6.410515691911261)
}; \addlegendentry{Perfect-CSI}
\end{axis}
\end{tikzpicture}}
\end{center}
\caption{Channel capacity as a function of the number of BS antennas for different values of the labeled dataset size $n$.}
\label{beam_selec}
\end{figure}
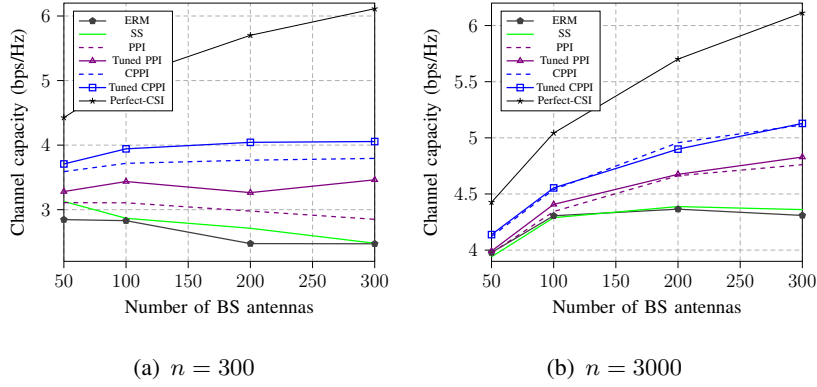

\subsection{Numerical Results}
We compare the performance of the proposed tuned CPPI approach to the benchmark schemes ERM, SS (with $\gamma=1$), PPI, tuned PPI, and CPPI described in Sec. \ref{problem_def} and Sec. \ref{PPI}. We also use as a reference the ideal case of perfect CSI.
% Moreover, we compare the \emph{location-based} beam alignment scheme~\cite{garcia2016location} where the AoAs and AoDs are computed based on the BS and UE regardless of the propagation environment. 
For CPPI and tuned CPPI, we train $K=8$ models using the procedure described in Section \ref{cp_sec}. As for PPI and tuned PPI, the labeled dataset is divided into two subsets of equal size, one used for training the prediction model $f(X)$ in \eqref{model_f_beam}, and the other for bias correction when computing the PPI and tuned PPI losses in \eqref{ppi_loss} and \eqref{tunedppi_loss}, respectively. The PPI model $f(X)$ and the CPPI models $\{f^{(k)}(X)\}$ are obtained using \eqref{model_f_beam}, where the CKM function $c(X)$ is implemented as a fully connected neural network with three hidden layers of size 128, 256, and 1024, respectively, with the LeakyRelu activation function \cite{xu2015empirical}.

We consider the same physical environment and dataset as in~\cite{zeng2021toward}. The dataset contains ground-truth multi-path channel information $Z $ in \eqref{Z_expression} generated by using the \emph{ray tracing} software Remcom Wireless Insite\footnote{https://www.remcom.com/wireless-insite-em-propagation-software}. { The BS is equipped with a $N_y\times N_z$ uniform planar array (UPA),} and the UE is equipped with a single antenna. Furthermore, Kronecker product-based beamforming
codebooks are employed~\cite{xie2013limited}. { As in~\cite{zeng2021toward}, the codebook is obtained based on the transmit array responses, with uniformly sampled
values over $\sin(\phi)$ and $\cos(\theta)$, having respective resolutions of $1/{(2N_y)}$ and $1/{(2N_z)}$.} The total number of samples available in the dataset is $38038$, from which $n$ samples are reserved as labeled data, and the remaining $N  =38038-n $ samples are considered as unlabeled data. 

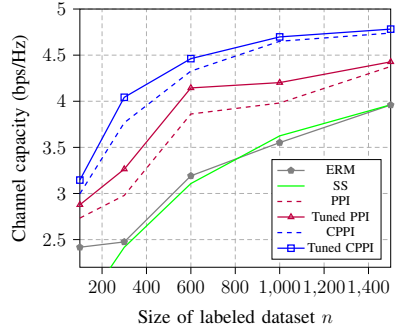
\begin{figure}[t]
\begin{center}
\begin{tikzpicture}[scale=0.6]
\begin{axis}[
tick align=outside,
tick pos=left,
x grid style={white!69.0196078431373!black},
xlabel={Size of labeled dataset $n$},
xmajorgrids,
xmin= 100, xmax=1500,
xtick style={color=black},
y grid style={white!69.0196078431373!black},
ylabel={Channel capacity (bps/Hz)},
ymajorgrids,
ymin=2.2, ymax=5,
ytick style={color=black},
grid=major,
scaled ticks = true,
legend pos = south east,
legend style={nodes={scale=0.7, transform shape}},
grid style=densely dashed,
]
\addplot  [thick, color = gray, mark = pentagon*, mark size = 2, mark repeat = 0, mark phase = 0]
coordinates {
(100, 2.417119392583258)(300, 2.4747949110130256)(600, 3.191475093797135)(1000, 3.5519291023289896)(1500, 3.95980650302989)(2200, 4.162599625266602)(3000, 4.364690547432143)
}; \addlegendentry{ERM}
\addplot  [thick, color = green, mark = none, mark size = 2, mark repeat = 0, mark phase = 0]
coordinates {
(100, 1.7092555509383554)(300, 2.4164460787246917)(600, 3.1077813742583142)(1000, 3.624658998412027)(1500, 3.9617063894958937)(2200, 4.329959635002709)(3000, 4.387841194681612)
% (100, 2.414908342237271)(300, 2.7120769123778374)(600, 2.7036674266479603)(1000, 3.1849197899299835)(1500, 3.019583904823556)(2200, 3.0947809858346935)(3000, 3.4783780432668454)
}; \addlegendentry{SS}
\addplot  [thick, color = purple, dashed , mark =none,  mark size = 2, mark repeat = 0, mark phase = 0]
coordinates {
(100, 2.7332404582444583)(300, 2.978899322031873)(600, 3.861775003812623)(1000, 3.98097260604248)(1500, 4.377017411335242)(2200, 4.651795488072257)(3000, 4.662660785457252)
}; \addlegendentry{PPI}
\addplot  [thick, color = purple,  mark = triangle,  mark size = 2, mark repeat = 0, mark phase = 0]
coordinates {
(100, 2.8775712367455353)(300, 3.2642934704791373)(600, 4.1445435612867305)(1000, 4.202231359463732)(1500, 4.4271424641937625)(2200, 4.660800962695074)(3000, 4.674837160560732)
}; \addlegendentry{Tuned PPI}
\addplot  [thick, color =blue, dashed , mark = none,  mark size = 2, mark repeat = 0, mark phase = 0]
coordinates {
(100, 2.9967508024053853)(300, 3.7669154054205753)(600, 4.326774334005395)(1000, 4.650579427459735)(1500, 4.739281595762133)(2200, 4.893123876398642)(3000, 4.956654767902216)
}; \addlegendentry{CPPI}
\addplot  [thick, color = blue, mark = square, mark size = 2, mark repeat = 0, mark phase = 0]
coordinates {
(100, 3.1454113920854856)(300, 4.04301454960707)(600, 4.462151282893996)(1000, 4.69770562438163)(1500, 4.7821103266259755)(2200, 4.90747249694001)(3000, 4.8993033686150556)
}; \addlegendentry{Tuned CPPI}
\end{axis}
\end{tikzpicture}
\end{center}
\caption{Channel capacity as a function of the size of labeled dataset $n$ when the number of BS antennas is fixed to $N^{TX} = 200$.}
\label{beam_selec_vsn}
\end{figure}

In \figref{beam_selec}, we report the performance in terms of the channel capacity $\log_2(1+\text{SNR})$, {  with $\text{SNR}$ defined as
\begin{equation}
\text{SNR} = \frac{P \left|u_Y^H Hw_Y\right|^2}{\sigma^2},
\end{equation}
where $P$ is the transmit power at the BS, $\sigma^2$ is the noise variance, and $Y$ is the index of the selected beam pair. The index $Y$ is evaluated as \eqref{linear_softmax_reg} for all schemes, except for the perfect CSI case, for which\eqref{beam_selection} is used instead.} We vary the number of transmit antennas at the BS while setting the size of the labeled dataset to $n=300$ and $n=3000$. { The size of the UPA array at the BS is such that $N_z=10$ and $N_y$ varies from 5 to 30}. It is observed that the proposed tuned CPPI outperforms all the benchmark schemes while exhibiting the same performance as CPPI when the number of labeled data points, $n$, is large enough. { Indeed, when $n$ is sufficiently high, the trained models, for both PPI and CPPI, are accurate, and thus the optimal tuning parameter $\lambda^*$ becomes close to 1. Consequently, tuned PPI and tuned CPPI perform in a way similar to PPI and CPPI, respectively.}
% This highlights the benefits of the cross-prediction-powered inference. 
% approach in general, where starting from a limited amount of labeled data, one can obtain significant performance gain by leveraging predictions on unlabeled data.

We also remark from the figure that conventional semi-supervised (SS) learning, which does not take into consideration the bias from the trained prediction models, fails to provide any gains by incorporating the unlabeled data. Finally, we note that PPI and tuned PPI, while offering some gain as compared to ERM, do not reach the performance of tuned CPPI.

% {  It is noted in \figref{beam_selec}(a) that the performance of most of the schemes does not increase with the number of antennas. We think the cause for that is the low size of the labeled data as evidenced by \figref{beam_selec}(b), where the rates increase with the number of BS antennas for higher labeled dataset size.}

To further elaborate on the impact of the number of labeled data points, $n$, we report the channel capacity as a function of the labeled dataset size $n$ in \figref{beam_selec_vsn}. { The number of BS antennas in this experiment is fixed to $N^{TX}=200$ with $N_y=20$ and $N_z=10$.} As seen, the proposed tuned PPI provides the best performance among all the benchmark schemes, with more significant gains for smaller values of the labeled dataset size $n$. In fact, in this regime, the trained CKM is not sufficiently accurate, and tuned CPPI, which can adapt to the quality of the trained models via the tuning parameter $\lambda$, yields better performance.

{ 
\subsection{Experiments with Neural Network Models}
\label{sim_nn_models}
Considering the same system model described in Sec. VIII.\ref{sys_model}, we present here experiments with neural network (NN) models, providing also comparison with MPl and with the proposed MCPPI. It is assumed that the labeled dataset contains the UE location $X$ and the corresponding beam index $Y$. Furthermore, models $\{f^{(k)}(\cdot)\}_{k=1}^K$ are NN classifiers predicting the beam index, which are implemented as
fully connected NNs with two hidden layers of size 128 and 256, respectively, with the Relu activation function. We use the same architecture for the target model parameterized by $\theta$, which is also referred to as the student model $S_\theta(\cdot)$ in Sec.\ref{CPPI_nn}.

We start by comparing tuned CPPI with the benchmark schemes ERM, SS, PPI, tuned PPI, and CPPI in \figref{beam_selec_nns}. Tuned CPPI outperforms CPPI and shows a significant gain compared to using labeled data alone as in ERM. As in \figref{beam_selec_vsn}, PPI and tuned PPI provide gains compared to ERM but are outperformed by CPPI and tuned CPPI. 

To further demonstrate the efficiency of the proposed schemes, we provide a comparison with the strong SSL benchmark MPL~\cite{pham2021meta}. In \figref{beam_selec_nns_an}, we plot the channel capacity as a function of the number of BS antennas for MPL, CPPI, tuned CPPI, and the proposed MCPPI. The results show that CPPI and tuned CPPI deliver slightly lower performance than MPL, indicating that they are competitive SSL schemes. Moreover, MCPPI is able to overcome this deficit and provides better performance than MPL.

\begin{figure}[ht]
\begin{center}
\begin{tikzpicture}[scale=0.6]
\begin{axis}[
tick align=outside,
tick pos=left,
x grid style={white!69.0196078431373!black},
xlabel={Size of labeled dataset $n$},
xmajorgrids,
xmin= 512, xmax=4100,
xtick style={color=black},
y grid style={white!69.0196078431373!black},
ylabel={Channel capacity (bps/Hz)},
ymajorgrids,
ymin=3.5, ymax=5.2,
ytick style={color=black},
grid=major,
scaled ticks = true,
legend pos = south east,
legend style={nodes={scale=0.6, transform shape}},
grid style=densely dashed,
]

\addplot  [semithick, color = red, mark = star, mark size = 2, mark repeat = 0, mark phase = 0]
coordinates {
(128, 2.514505359522904)(256, 2.961611773529421)(512, 3.4743600451567844)(1024, 3.7288443195443928)(2048, 4.003198497178796)(4096, 4.443089397969243)
}; \addlegendentry{ERM}

\addplot  [thick, color = green, mark = none, mark size = 2, mark repeat = 0, mark phase = 0]
coordinates {
(128, 1.939229513428491)(256, 2.218801517818342)(512, 2.825604075466892)(1024, 3.6582947308801166)(2048, 4.047121394096558)(4096, 4.409081553728858)
}; \addlegendentry{SS}
\addplot  [thick, color = violet, dashed , mark =none,  mark size = 2, mark repeat = 0, mark phase = 0]
coordinates {
(128, 2.879609341957804)(256, 3.3686166749777158)(512, 3.7546085365760127)(1024, 4.430991417880725)(2048, 4.616340543165606)(4096, 4.83076669857406)
}; \addlegendentry{PPI}
\addplot  [thick, color = violet,  mark = triangle,  mark size = 2, mark repeat = 0, mark phase = 0]
coordinates {
(128, 2.8718675582774065)(256, 3.3895610809050414)(512, 3.6760500342169613)(1024, 4.4109511203354606)(2048, 4.6578578646712865)(4096, 4.840696840513544)
}; \addlegendentry{Tuned PPI}

\addplot  [thick, color =blue, dashed , mark = none,  mark size = 2, mark repeat = 0, mark phase = 0]
coordinates {
(128,3.55493712)(256 4.06670422)(512, 4.47434381 ) (1024,4.88169697)(2048,5.06650802)(4096,5.10397698)

}; \addlegendentry{CPPI}

\addplot  [thick, color = blue, mark = square, mark size = 2, mark repeat = 0, mark phase = 0]
coordinates {
(128,3.6353396)(256,4.09319052)(512,4.52715014)(1024,4.9655461)(2048,5.12293704)(4096,5.17871482)
}; \addlegendentry{Tuned CPPI}

\end{axis}
\end{tikzpicture}
\end{center}
\caption{Channel capacity as a function of the size of labeled dataset $n$ when the number of BS antennas is fixed to $N^{TX} = 200$.}
\label{beam_selec_nns}
\end{figure}
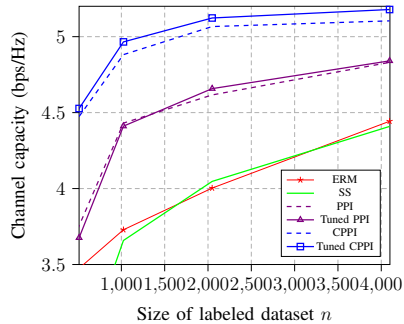

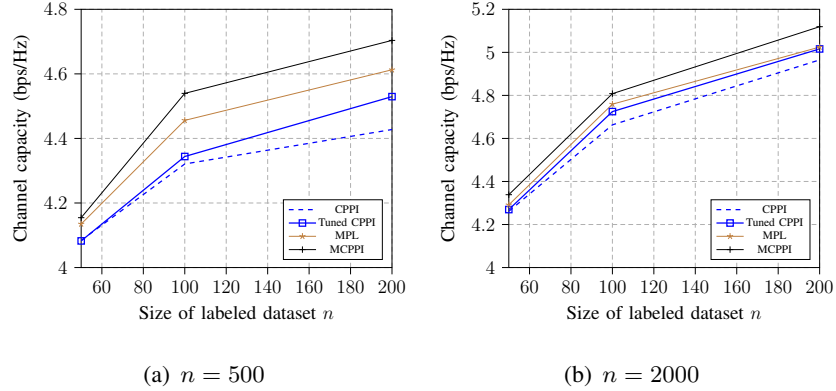
\begin{figure}[ht]
\begin{center}
\subfigure[$n = 500$]{
\begin{tikzpicture}[scale=0.6]
\begin{axis}[
tick align=outside,
tick pos=left,
x grid style={white!69.0196078431373!black},
xlabel={Size of labeled dataset $n$},
xmajorgrids,
xmin= 50, xmax=200,
xtick style={color=black},
y grid style={white!69.0196078431373!black},
ylabel={Channel capacity (bps/Hz)},
ymajorgrids,
ymin=4, ymax=4.8,
ytick style={color=black},
grid=major,
scaled ticks = true,
legend pos = south east,
legend style={nodes={scale=0.6, transform shape}},
grid style=densely dashed,
]
\addplot  [thick, color =blue, dashed , mark = none,  mark size = 2, mark repeat = 0, mark phase = 0]
coordinates {
(50,4.08410029)(100, 4.32087662)(200, 4.42758192 )
}; \addlegendentry{CPPI}

\addplot  [thick, color = blue, mark = square, mark size = 2, mark repeat = 0, mark phase = 0]
coordinates {

(50,4.08253645)(100 ,4.34401383)(200, 4.52949144 )
}; \addlegendentry{Tuned CPPI}

\addplot  [semithick, color = brown, mark = star, mark size = 2, mark repeat = 0, mark phase = 0]
coordinates {
(50, 4.13484012)(100, 4.45593277 )(200, 4.61269783)
}; \addlegendentry{MPL}

\addplot  [semithick, color = black, mark = +, mark size = 2, mark repeat = 0, mark phase = 0]
coordinates {
(50, 4.15447439 )(100,  4.5395619)(200, 4.70367368)
}; \addlegendentry{MCPPI}

\end{axis}
\end{tikzpicture}}
\subfigure[$n = 2000$]{
\begin{tikzpicture}[scale=0.6]
\begin{axis}[
tick align=outside,
tick pos=left,
x grid style={white!69.0196078431373!black},
xlabel={Size of labeled dataset $n$},
xmajorgrids,
xmin= 50, xmax=200,
xtick style={color=black},
y grid style={white!69.0196078431373!black},
ylabel={Channel capacity (bps/Hz)},
ymajorgrids,
ymin=4, ymax=5.2,
ytick style={color=black},
grid=major,
scaled ticks = true,
legend pos = south east,
legend style={nodes={scale=0.6, transform shape}},
grid style=densely dashed,
]
\addplot  [thick, color =blue, dashed , mark = none,  mark size = 2, mark repeat = 0, mark phase = 0]
coordinates {
(50, 4.261288513906704)(100, 4.663109137175012)(200, 4.964895015528836)
}; \addlegendentry{CPPI}

\addplot  [thick, color = blue, mark = square, mark size = 2, mark repeat = 0, mark phase = 0]
coordinates {
(50, 4.269412882720594)(100, 4.724744304426871)(200, 5.015903796154096)
}; \addlegendentry{Tuned CPPI}

\addplot  [semithick, color = brown, mark = star, mark size = 2, mark repeat = 0, mark phase = 0]
coordinates {
(50, 4.291356479029062)(100, 4.75914603449723)(200, 5.024543033212285)
}; \addlegendentry{MPL}

\addplot  [semithick, color = black, mark = +, mark size = 2, mark repeat = 0, mark phase = 0]
coordinates {
(50, 4.3385862725976105)(100, 4.808051385627203)(200, 5.119252466496168)
}; \addlegendentry{MCPPI}

\end{axis}
\end{tikzpicture}}
\end{center}
\caption{Channel capacity as a function of the number of BS antennas when the size of labeled data is fixed to $n =500$ and $n=2000$.}
\label{beam_selec_nns_an}
\end{figure}

}

% \section{Experiments}
% \label{Experiments}

\section{Conclusion}
\label{conclusion}
In this work, we have investigated the potential benefits of prediction-powered inference (PPI) for wireless systems. PPI leverages predictions generated by ML models to augment labeled data, while carefully considering their inherent prediction bias. PPI and its variant tuned PPI assume the existence of a pre-trained prediction model, while cross PPI (CPPI) alleviates this assumption by leveraging labeled data for the tasks of training prediction models and of correcting the inherent bias.

We have specifically explored two potential applications of the framework within wireless systems: CKM-based beam alignment and RSSI-based indoor localization. In addition, we have introduced a novel variant of PPI, tuned CPPI, which endows CPPI with the capacity to tailor based on the accuracy of the trained prediction models. We have compared the performance of all PPI variants with ERM, a baseline scheme reliant solely on labeled data, and with conventional semi-supervised learning, which overlooks the prediction bias. Numerical results show the superiority of the proposed tuned CPPI scheme over PPI variants and classical pseudo-labeling schemes. We have also investigated the integration of CPPI with meta pseudo labeling, a strong SSL benchmark. This resulted in a new scheme outperforming both MPL and tuned CPPI.

% We believe that the PPI framework including the proposed tuned CPPI scheme could be applied in many other scenarios in wireless systems. Investigating these potential scenarios is an envisaged future direction. the combination of tuned PPI and tuned CPPI could be also an interesting direction to pursue
% where both the predictions from a given model $f(.)$ and models trained using the available labeled data could be combined with different tuning parameters to further enhance the performance.

PPI schemes can benefit any wireless application in which labeled data are costly to obtain, but unlabeled data are available. Identifying and investigating these applications represents a promising direction for future research. Of particular interest are applications that leverage digital twins of the physical system to augment the datasets~\cite{jiang2023digital,ruah2023calibrating,luo2024digital}. At a methodological level, extensions of the proposed tuned CPPI to active learning methods~\cite{zrnic2024active} could be addressed by further research.

% Additionally, exploring the synergy between the tuned PPI and CPPI schemes could yield significant improvements. By combining predictions from a given model with predictions from models trained using the available labeled data, and integrating them with different tuning parameters, the performance could be further enhanced.

\section*{Appendix A\\Bootstrap-Based Estimation}
\label{appen_A}
Here, we describe the bootstrap approach to estimate the covariance matrix $\var(\nabla\ell^{\bar f}_{\theta^\star})$ and the cross-covariance matrix $\cov(\nabla\ell_{\theta^\star},\nabla\ell^{\bar f}_{\theta^\star})$. The approach consists of simulating several runs of the training process and averaging their results. For each run $b\in\{1,\cdots, B\}$, we sample $n-n/K$ data points from the labeled dataset, denoted by $\mathcal{I}_b$, and train a model $f^{(b)}(X)$ using these data points. The predictions from the trained models $\{f^{(b)}(X)\}_{b=1}^B$ are leveraged to estimate $\var(\nabla\ell^{\bar f}_{\theta^\star})$ and $\cov(\nabla\ell_{\theta^\star},\nabla\ell^{\bar f}_{\theta^\star})$. 

The first covariance $\var(\nabla\ell^{\bar f}_{\theta^\star})$ can be estimated using the unlabeled dataset as follows
\begin{align}
\widehat{\var}(\nabla\ell^{\bar f}_{\hat\theta}) = \widehat{\var}\left( \nabla \ell_{\hat\theta^{}_{}}(\tilde X_i,\bar{f}(\tilde X_i)); \ i\in[N] \right),
\label{bootstrap_cov_unlabeled}
\end{align}
where $\bar{f}(\tilde X_i) = \frac{1}{b}\sum_{b=1}^Bf^{(b)}(\tilde X_i)$, and $\widehat{\var}(v_i; \ i\in \mathcal{I} )$ denotes the empirical estimate of the covariance matrix of vectors $v_i$ using data samples with indices $i\in \mathcal{I} $, i.e.,
$$
\widehat{\var}(v_i; \ i\in \mathcal{I} ) = \frac{1}{|\mathcal{I}|}\sum_{i\in \mathcal{I}} (v_i-\bar v)(v_i-\bar v)^T,
$$
with $ \bar v  = \frac{1}{|\mathcal{I}|}\sum_{i\in \mathcal{I}} v_i$. As for the cross-covariance $\cov(\nabla\ell_{\theta^\star},\nabla\ell^{\bar f}_{\theta^\star})$, the labeled dataset is used as
\begin{align}
\widehat{\cov}(\nabla\ell_{\hat\theta},\nabla\ell^{\bar f}_{\hat\theta}) &= \widehat{\cov}\bigg(\nabla\ell_{\hat\theta}(X_i,Y_i),\nabla\ell^{}_{\hat\theta}(X_i, f^{(b)}(X_i)); \quad \ i \in  \{1,\cdots,n\}\backslash I_b, \ b\in \{1,\cdots,B\}\bigg),
\label{bootstrap_cross_cov}
\end{align}
where $\widehat{\cov}(u_i,v_i; \ i\in \mathcal{I} )$ denotes the empirical estimate of the covariance matrix between vectors $u_i$ and $v_i$ using data samples with indices $i\in \mathcal{I} $, i.e.,
$$
\widehat{\cov}(u_i,v_i; \ i\in \mathcal{I} ) =  \frac{1}{|\mathcal{I}|}\sum_{i\in \mathcal{I}} (u_i-\bar u)(v_i-\bar v)^T,
$$
with $\bar u  = \frac{1}{|\mathcal{I}|}\sum_{i\in \mathcal{I}} u_i$ and $\bar v  = \frac{1}{|\mathcal{I}|}\sum_{i\in \mathcal{I}} v_i$. Note that the bootstrapped models are not averaged when estimating ${\cov}(\nabla\ell_{\theta^\star},\nabla\ell^{\bar f}_{\theta^\star})$ in \eqref{bootstrap_cross_cov} to ensure that each sample from the labeled dataset used 
to compute $\nabla\ell_{\hat\theta}(X_i,Y_i)$ is independent of the samples $f^{(b)}$ was trained on.

{ 
\section*{Appendix B\\Derivation of the Teacher Models Update Rule for MCPPI}
\label{appen_B}

\begin{algorithm}[!t]
 % \SetAlgoNoLine
  
    \KwIn{Labeled dataset $\mathcal{D}$,  unlabeled dataset $\tilde{\mathcal{D}}$}
    \KwOut{Optimized parameters $\theta^\star$}
\For{ gradient step $g$ in $\{1,\cdots,G\}$}{
Draw a batch of labeled data $(X,Y)$ from $\mathcal{D}$ and a batch of unlabeled data $\tilde X$ from $\tilde{\mathcal{D}}$\\
Sample $\hat{y}^{(k)}_u \sim f^{(k)}_{\phi^k}(\tilde X)$ and $\hat{y}^{(k)}_l \sim f^{(k)}_{\phi^k}( X)$, for $k \in \{1,\cdots,K\}$\\
Update the student model parameters using \eqref{metacppi_theta_approx}\\
\For{ $k$ in $\{1,\cdots,K\}$}{
Update parameters $\phi_k$ by performing a gradient step on the loss \eqref{metacppi_teacher_loss_approx} using~\eqref{J_approx}
}
}  
\caption{Meta-CPPI}
\label{mcppi_alg}  \vspace{0mm}
\end{algorithm}

Focusing on classification problems, we set the loss function as $\ell_\theta(X,Y)= {\rm CE}(Y, S_\theta(X))$, where ${\rm CE}(\cdot,\cdot)$ denotes the cross-entropy loss. We further assume that each of the teacher models $\{f^{(k)}_{\phi^k}(\cdot)\}_{k=1}^K$ produces a probability distribution over the labels, typically using a softmax output layer. With some abuse of notation, we write as $f^{(k)}_{\phi^k}(\cdot)$ the output probability distribution.

% Note that at each gradient step, $\theta$ is treated as fixed and $\theta^\star(\{\phi^k\}_{k=1}^K$ in \eqref{metacppi_theta_approx} is estimated as the average over batches. 

Moreover, for simplicity of notation, we use $\tilde{X}$ to represent a batch of unlabeled samples, and $(X, Y)$ to denote a batch of labeled data. The required gradient of the second term of the loss \eqref{metacppi_teacher_loss_approx} is as follows
\begin{align}
J_k=   \nabla_{\phi^k}\ex_{(X,Y) \sim \mathcal{D}} \left[ {\rm CE}\left( Y, S_{\theta^\star(\{\phi^k\}_{k=1}^K)}(X)\right)\right].
 \label{J_a}
\end{align}
Following the same steps as in \cite[Appendix A]{pham2021meta}, an approximation of the gradient $J_k$ is obtained as
\begin{align}
J_k&\approx \lambda \eta_S \left(\left(\nabla_{\theta'}{\rm CE}(S(X), Y) \right)^T \nabla_{\theta}{\rm CE}(S(\tilde X), \hat{y}^{(k)}_u)\right) \cdot \nabla_{\phi^k}{\rm CE}(f^{(k)}(\tilde X), \hat{y}^{(k)}_u)\nonumber \\
& \quad - \lambda \eta_S \left(\left(\nabla_{\theta'}{\rm CE}(S(X), Y) \right)^T \nabla_{\theta}{\rm CE}(S(X), \hat{y}^{(k)}_l)\right)\cdot\nabla_{\phi^k}{\rm CE}(f^{(k)}(X), \hat{y}^{(k)}_l), 
 \label{J_approx}
\end{align}
 with $\theta'$ being a batch estimate of $\theta^\star(\{\phi^k\}_{k=1}^K)$ in \eqref{metacppi_theta_approx}, and $\hat{y}^{(k)}_u$ (respectively, $\hat{y}^{(k)}_l$) being sampled as $\hat{y}^{(k)}_u \sim f^{(k)}_{\phi^k}(\tilde X)$ (respectively, as $\hat{y}^{(k)}_l \sim f^{(k)}_{\phi^k}(X)$). $\lambda$ is the tuning parameter of the tuned CPPI loss of the student model in~\eqref{metacppi_student_loss}. The optimal tuning parameter $\lambda^\star$ is estimated using \eqref{opt_lam_mean} for MCPPI. Note that the first term in \eqref{J_approx} is the same as~\cite[Equation (12)]{pham2021meta}, while the second term follows by considering the tuned CPPI loss \eqref{metacppi_student_loss}. The pseudo-code of the proposed MCPPI is given in Algorithm \ref{mcppi_alg}.

}

\bibliographystyle{IEEEtran}
\bibliography{references}
\end{document}